\pdfoutput=1
\documentclass[11pt,a4paper]{article}

\usepackage[lmargin=1.1in,rmargin=1.1in,bottom=1.3in,top=1.3in,
twoside=False]{geometry}

\usepackage{microtype}
\usepackage{amsmath,mathtools}  
\usepackage{amsfonts}
\usepackage{multirow}
\usepackage{float}
\usepackage{amsthm}
\usepackage{pgfplotstable}
\usepackage{pgfplots}
\usepackage{comment}
\usepackage[ruled,vlined,commentsnumbered,noend]{algorithm2e}
\usepackage[noend]{algorithmic}
\usepackage[title]{appendix}
\usepackage[square,sort,comma,numbers]{natbib}
\usepackage{hyperref}
\usepackage{cleveref}
\usepackage{paralist}
\usepackage{amssymb}

\newcommand{\N}{\mathbb{N}}

\newcommand{\sizeof}[1]{\left|#1\right|}

\newcommand{\coloring}{{\sc $P_k$-free coloring problem}\xspace}

\newcommand{\nph}{$\mathsf{NP}$-hard\xspace}

\newcommand{\sd}[1]{$S(#1)$\xspace}

\newcommand{\wasted}{wasted edge\xspace}
\newcommand{\wasteds}{wasted edges\xspace}

\newenvironment{ClaimProof}[1]{\par\noindent{\textit{Proof}}\space#1}{$\blacksquare$}
\newenvironment{repeathm}[1]{\par\addvspace{3mm}\noindent\textbf
  {\Cref{#1}.\;}}{\par\addvspace{3mm}}

\newcommand{\footremember}[2]{%
    \footnote{#2}
    \newcounter{#1}
    \setcounter{#1}{\value{footnote}}%
}
\newcommand{\footrecall}[1]{%
    \footnotemark[\value{#1}]%
}

\title{Complexity of Computing the Anti-Ramsey Numbers for Paths}

\author{Saeed Akhoondian Amiri\footremember{University of Cologne}{University of Cologne, Cologne, Germany \texttt{amiri@informatik.uni-koeln.de}}
\quad Alexandru Popa\footremember{Bucharest}{University of Bucharest and National Institute of Research and Development in Informatics, Bucharest, Romania \texttt{alexandru.popa@fmi.unibuc.ro}}\quad Mohammad Roghani\footremember{MPI}{MPI for Informatics, Saarland Informatics Campus, Saarbr\"ucken, Germany \texttt{$\{$mohammadroghani43, rsoltani97$\}$@gmail.com} \texttt{$\{$gshahkar, hovahidi$\}$@mpi-inf.mpg.de}} \footremember{Sharif}{Sharif University of Technology, Tehran, Iran }\\\quad Golnoosh Shahkarami\footrecall{MPI}  \footremember{GradSchool}{Saarbr\"ucken Graduate School of Computer Science, Saarbr\"ucken, Germany}\quad Reza Soltani\footrecall{MPI} \footrecall{Sharif}\quad Hossein Vahidi \footrecall{MPI} \footrecall{GradSchool}}

\newtheorem{theorem}{Theorem}
\newtheorem{problem}[theorem]{Problem}
\newtheorem{observation}[theorem]{Observation}
 \newtheorem{definition}[theorem]{Definition}
\newtheorem{lemma}[theorem]{Lemma}
\newtheorem{cclaim}{Claim}[theorem]

\newtheorem{claim}{Claim}[theorem]

\newcommand{\myconst}{(f_k+1)}
\newcommand{\myconsttt}{4}
\DeclareMathOperator*{\argmax}{argmax}

\begin{document}
\date{}
\maketitle

\begin{abstract}
The anti-Ramsey numbers are a fundamental notion in graph theory,
introduced in 1978, by Erd\" os, Simonovits and S\' os.  For given
graphs $G$ and $H$ the \emph{anti-Ramsey number} $\textrm{ar}(G,H)$ is
defined to be the maximum number $k$ such that there exists an
assignment of $k$ colors to the edges of $G$ in which every copy of
$H$ in $G$ has at least two edges with the same color.

Usually, combinatorists study extremal values of anti-Ramsey numbers
for various classes of graphs. There are works on the computational
complexity of the problem when $H$ is a star. Along this line of
research, we study the
complexity of computing the anti-Ramsey number $\textrm{ar}(G,P_k)$, where $P_k$
is a path of length $k$. First, we observe that when $k$ is close
to $n$, the problem is hard; hence, the challenging part is the
computational complexity of the problem when $k$ is a fixed constant.

We provide a characterization of the problem for paths of
constant length.
Our first main contribution is to prove that computing
$\textrm{ar}(G,P_k)$ for every integer $k>2$ is NP-hard. We obtain
this by providing several structural properties of such coloring in
graphs.  We investigate further and show that approximating
 $\textrm{ar}(G,P_3)$ to a factor of $n^{-1/2 - \epsilon}$ is hard
 already in $3$-partite graphs, unless $P{}=
 {}NP$. We also study the exact complexity of the precolored
version and show that there is no subexponential algorithm for the
problem unless ETH fails for any fixed constant $k$.

Given the hardness of approximation and parametrization of the problem, it is natural to study the problem on restricted graph families. Along this line,
we first introduce the notion of color connected
coloring, and, employing this structural property, we obtain a
linear time algorithm to compute $\textrm{ar}(G,P_k)$, for every integer
$k$, when the host graph, $G$, is a tree. We have introduced several
techniques in our algorithm that we believe might be helpful in
providing  approximation algorithms for other
restricted families of graphs.
\end{abstract}

\section{Introduction}

For given graphs $G$ and $H$, the \emph{anti-Ramsey number} $\textrm{ar}(G,H)$ is defined to be the maximum number $k$ such that there exists an assignment of $k$ colors to the edges of $G$ in which every copy of $H$ in $G$ has at least two edges with the same color. Classically, the graph $G$ is a large complete graph and the graph $H$ is from a particular graph class. 

The study of anti-Ramsey numbers was initiated by Erd\" os, Simonovits and S\' os in 1975~\cite{ErdosSS75}. Since then, there have been a large number of papers on the topic. There are papers that study the case when $G = K_n$ and $H$ is a: cycle, e.g.,~\cite{ErdosSS75,Montellano-Ballesteros2005,AxenovichJK04}, tree,  e.g.,~\cite{JiangW04,Jiang2002}, clique, e.g.,~\cite{FriezeR93,ErdosSS75,BlokhuisFGR01}, matching, e.g.,~\cite{Schiermeyer04,ChenLT09,HaasY12} and others,  e.g.,~\cite{ErdosSS75,AxenovichJ04}. 

The anti-Ramsey numbers are connected with the rainbow number~\cite{FujitaMO14} $rb(G,H)$, which is defined as the minimum number $k$ such that in \emph{any} coloring of  the edges of $G$ with $k$ colors, there exists a rainbow copy of $H$. Thus, $ar(G,H) = rb(G,H) - 1$. We call a coloring without a rainbow copy of $H$, an \emph{$H$-free coloring}.

Various combinatorial works studied the case when $H$ is a path or a cycle. For instance, the work of Simonovits and Sos~\cite{SimonovitsS1984} shows that there exists a constant $t$ such that for a sufficiently long path $ar(K_n, P_t)\in O(t\cdot n)$. The combinatorial analysis of the problem is extremely difficult when instead of $K_n$ we use an arbitrary graph as the \emph{host} graph. For a more detailed exposition of the combinatorial results on anti-Ramsey numbers, we refer the reader to the following surveys:~\cite{Ingo2007,FujitaMO14}. 

Besides the extremal results, the anti-Ramsey numbers have been studied from the computational point of view in several papers.  The anti-Ramsey numbers when $G$ is an arbitrary graph was studied for the case when $H$ is a star. The problem was introduced by  Feng et al.~\cite{FengZQW07,FengCZ08,FengZW09}, motivated by applications in wireless mesh networks and was termed the \emph{maximum edge $q$-coloring}.

They provide a $2$-approximation algorithm for $q=2$ and a $(1+\frac{4q-2}{3q^2-5q+2})$-approximation for $q > 2$. They show that the problem is solvable in polynomial time for trees and complete graphs in the case $q=2$. Later, Adamaszek and Popa~\cite{AdamaszekP10} show that the problem is APX-hard and present a $5/3$-approximation algorithm for graphs with a perfect matching.  For more results related to the maximum edge $q$-coloring, the reader can refer to~\cite{AdamaszekP16}.

To improve our understanding on such problems, we continue the recent line of study of the computational complexity of the problem.
Similar to previous works we restrict $H$ to a basic class of graphs,  paths. We let $G$  be either an arbitrary graph or a restricted family of graphs such as trees or bipartite graphs. We provide a big picture on what is tractable and what is not tractable when we are dealing with anti-Ramsey numbers on paths. Namely we prove the following.
 

\subsection*{Our Results}
 \begin{enumerate}
 \item First, we show that computing the value of $ar(G,P_k)$ is \nph for every $k > 2$ via a reduction from the maximum independent set problem. Namely, we prove the following theorem.
 
\begin{theorem}\label{thm:hardness}
	For every $k > 2$, \coloring is \nph.
\end{theorem}

The above theorem basically states that there is no XP algorithm, parameterized by $k$, for the problem unless $P{}={}NP$. The reduction is multi stage: firstly we distinguish between the odd and even values of $k$. Then for each parity of $k$, given an instance of independent set, we construct an auxiliary graph and prove several structural lemmas on that graph to establish a one to one mapping between the maximum independent set in the original graph and the maximum anti-Ramsey coloring on the auxiliary graph. By a more careful analysis of the above proof for the special case of $k=3$, we show the problem is inapproximable by a factor $n^{- 1/2 - \epsilon}$, even on $3$-partite graphs, unless $P{}={}NP$.

Given the hardness of the problem, it is natural to investigate what would be the best exponential algorithm for the problem. We study the running time of the exact algorithm for a slight variant of the problem, namely, Precolored $P_k$-free coloring. We prove that the problem does not admit an exact algorithm with running time $2^{o(\sizeof{E(G)})}$ assuming ETH. 
\begin{theorem}\label{thm:ethprecolored}
There is no $2^{o(|E(G)|)}$ algorithm for Precolored $ar(G,P_k)$, for any fixed $k$, unless ETH fails.
\end{theorem}

To obtain such a reduction, we provide a graph construction with low edge density gadgets. This is unlike standard hardness proofs where it is possible to blow up the graph by any polynomially bounded size.

 \item Given the above hardness results, even for small values of $k$, it is natural to explore the tractability of problem when the host graph has a nice structural property. We first introduce a generic algorithmic idea, of \emph{color connected coloring} and we exploit this to develop a linear time algorithm on trees. 
\begin{theorem}\label{thm:TreePoly}
	For a tree $T$, there is an exact linear time algorithm
        that computes $ar(T,P_k)$ for every constant integer $k$; the algorithm runs in time $O(|V(T)| k^4)$.
\end{theorem}
Our algorithm is based on dynamic programming on trees, however, unlike most problems in trees, this one is not that straightforward and we employed several techniques to solve the problem.
There are known combinatorial results for cycles of length three on outerplanar graphs~\cite{DBLP:journals/dmgt/GoddardX16} and the algorithm for trees for $3$-consecutive coloring of~\cite{BujtasSTDL12}. Our algorithm is independent of the latter; however, if we set $k=3$ our algorithm solves the aforementioned problem, while the other direction does not work.
\end{enumerate}

The paper is organized as follows. In~\Cref{sec:preliminaries}, we introduce preliminaries. Then, we prove the NP-hardness of computing $ar(G,P_k)$ in \Cref{apx:hardness_pk} and next, we show the hardness of inapproximability for $P_3$-free coloring. In~\Cref{apx:finegrain} we show the exact complexity result for Precolored $P_3$-free coloring. In~\Cref{apx:approximation_trees}, we provide an exact polynomial time algorithm for trees. Finally, in~\Cref{sec:conclusions}, we summarize the results and present directions for future work.

\section{Preliminaries}\label{sec:preliminaries}

We use $\N$ to denote the set of natural numbers and we write $[n]$ to
denote the set $\{1,\ldots,n\}$. We refer the reader 
to~\cite{DBLP:books/daglib/0030488} for basic notions related to graph
theory. All the graphs considered in this paper are simple and undirected.

Let $G$ be a graph, we write $V(G)$ for its vertices and $E(G)$ for its edges. For any vertex $v\in V(G)$ we define $N(v)=\bigl\{u\in V(G) \mid \{u,v\}\in E(G) \bigr\}$ to be the open neighborhood of $v$, and $N[v]= N(v) \cup \{v\}$ as its closed neighborhood. Similarly for any subset of vertices $A \in V(G)$ we define $N[A]=\bigcup_{v\in A} N[v]$, and $N(A)=N[A]\setminus A$.
For $k \in \N^+ $ we denote by $P_k$  a path with
$k+1$ vertices. The length of $P_k$ is $k$, the number of its edges.
 Also let $p$ be a $P_k$, depending on the context we may write $p=(e_1, \ldots, e_k)$ where $e_i\in E(p)$ or $p=(v_1, \ldots, v_{k+1})$ where $v_i \in V(p)$ to describe a path. 

\begin{definition}[Coloring] 
Given an undirected graph $G = (V,E)$, a coloring of the edges of $G$ is a function $c : E \to \mathbb{N}$. Similarly for any subset $A \subseteq E$ we define $c(A)=\bigcup_{e\in A}c(e)$.
\end{definition}

We call a coloring of the edges of a graph $G$ a \emph{rainbow
  coloring} if for every pair of edges $e\neq e'\in E$ we have
$c(e)\neq c(e')$. Let $G,H$ be two graphs, an edge coloring $c$ of $G$ is
\emph{$H$-free coloring} if there is no rainbow subgraph of $G$ isomorphic to
$H$. We denote the number of distinct colors used in $c$ by $c_{G,H}$.
Let $\mathcal{C}$ be the set of all $H$-free colorings of $G$. 
The anti-Ramsey number of $G$ is $ar(G,H)=\max_{c \in \mathcal{C}} c_{G,H}$. We observe that if $k$ is part of the input, then the problem of computing $ar(G, P_k)$ is at least as hard as finding a Hamiltonian path.

\begin{observation}\label[lemma]{lem:largek}
Computing $ar(G,P_{\sizeof{V(G)}-1})$ is NP-hard.
\end{observation}
\begin{proof} $ar(G,P_{|V(G)|-1}) = |E|$ if and only if $G$ does not have a Hamiltonian Path. \end{proof}
In the above we can replace Hamiltonian Path in the proof with longest path and in addition use the length of this path as parameter to prove the hardness for large values of $k$.

\section{Hardness of $P_k$ Anti-Ramsey Coloring}
\label{apx:hardness_pk}

\medskip

In this section for every $k > 2$, we prove the hardness by a
reduction from the maximum independent set (MIS) problem. 

\textbf{Proof Sketch: }
We construct a new graph $G'$ from a graph $G$ such that from a
maximum $P_k$-free coloring of $G'$, we can derive the size of the
maximum independent set of $G$. To obtain the desired result, we divide
the problem into three subproblems. We use the reduction with
different approaches for \begin{inparaenum}
\item $k = 4$, \item every even $k > 4$, \item every odd $k
> 1$. \end{inparaenum}

Roughly speaking, we replace every vertex
and edge with specific gadgets; this depends on the parity of
$k$. Afterward, in each case, intuitively,
we prove that if a vertex belongs to an independent set, its
corresponding gadget can be colored with more distinct colors than a
vertex that does not belong to an independent set. On the other hand,
for each case, we design edge gadgets such that their coloring can be
(almost) fixed in advance, despite the choice of
colors for the vertex gadgets. The crucial part of the proofs lies in the analysis of a
structure of the maximum $P_k$-free coloring of $G'$ and, exploiting the dependency
between vertex gadgets.

 \subsection*{Hardness of the Problem for Odd $k>1$}
\noindent 
\textbf{Assumption I: } In this part we assume $k>1$ is an odd integer.

In the following, we first present an upper bound on the number of
colors when the graph $H$ is a path. For certain technical reasons
that we will see in the proofs, we define a constant $c_k$ depending
only on $k$ with a particular lower bound.

\begin{lemma}\label[lemma]{lem:upperbound}
$ar(G,P_k)\leq c_k|V(G)|$ for some $c_k\in \Theta(k\sqrt{\log k})$ and $c_k> 3k\sqrt{\log k}$.
\end{lemma}

\begin{proof}
    Let $c$ be a $P_k$-free coloring of $G$
    with the maximum number of colors; we take the maximum size set
    of edges of distinct colors w.r.t. $c$. The resulting graph has no
    $P_k$ as a subgraph and hence it does not have any $P_k$ as a minor so
    by Mader's theorem~\cite{Mader1967,DBLP:books/daglib/0030488}
        it has at most $c'_k|V(G)|$ edges where $c'_k \in
        O(k\sqrt{\log k})$, we set $c_k=\max\{c'_k,1+ 3k\lceil\sqrt{\log k}\rceil\}$. 
\end{proof}

\noindent 
\textbf{Assumption II: }In this section, $c_k$ is what we
  used in~\Cref{lem:upperbound}. Whenever we write $I$ it means a
  maximum independent set in the graph $G$.

\medskip
Given an undirected graph $G$, we construct a graph $G'$ as follows:

\begin{enumerate}
	\item For each $ v \in V(G) $ we introduce two new
	vertices $s_v, t_v \in V(G')$ and $\myconst c_k
	|V(G)|$ internally disjoint paths of length $k-1$, $\mathcal{P}^v = \{ P^v_1, \ldots,
	P^v_{\myconst c_k |V(G)|} \}$, connecting $s_v$ to
        $t_v$. Later in~\Cref{lem:3colors} we determine the value of $f_k$.
	\item For each edge $\{v, u\} \in E(G)$, add $4$ new edges in $E(G')$:
	$\{s_v, t_u\}$, $\{t_v, s_u\}$, $\{t_v, t_u\}$, $\{s_v,
        s_u\}$. Let us define the union of all such edges in the entire graph $G'$ as $E^s_t$, more formally $E^s_t=\bigcup_{\{u,v\}\in E(G)}\{\{s_v, t_u\}, \{t_v, s_u\}, \{t_v, t_u\}, \{s_v, s_u\}\}$.
\end{enumerate}

An edge coloring is \emph{valid} if it is
a $P_k$-free coloring. We start by
providing some lemmas and observations on the structure of valid
colorings of $G'$ to establish a connection between such a coloring
and an independent set in $G$. 

\begin{lemma}\label[lemma]{lem:fewcolors}
	In any $P_k$-free coloring of $G'$ the edges in $E^s_t$ will receive
	at most $2 c_k |V(G)|$ distinct colors.
\end{lemma}
\begin{proof}
    The subgraph of $G'$ induced on endpoints of edges in $E^s_t$ has exactly $2|V(G)|$ vertices hence the lemma follows from~\Cref{lem:upperbound}.
\end{proof}

\begin{lemma}\label[lemma]{lem:valid C4 coloring}
    If $G$ is a cycle of length $2 (k-1)$ then $ar(G,P_k)=2 (k-2)$.
\end{lemma}

\begin{proof}
    First of all, we provide a coloring scheme for a cycle of
        length $2 (k-1)$ with $2 (k-2)$ distinct
        colors. Consider two vertices $s,v$ of this cycle which are within
        distance $k-1$ from each other. There are two internally vertex disjoint paths
        $P$ and $P'$ each of length $k-1$ between $s$ and $v$. Recall
        that $k>1$ is an odd number hence $k-1=2t, t>0$. Let suppose
        the edges of $P$ and $P'$ are $e_1,\ldots,e_{2t}$ and $e'_1,\ldots,e'_{2t}$
        respectively, w.r.t.\ their order of appearance from
        $s$ to $t$. We define a coloring function $c$ as follows. 
    
    \begin{align*}
    c\colon\begin{cases}
    c(e_i)=i,c(e'_i)=i+k-1, \hspace{10mm} \text{ if } i\neq t \text{ and }i\neq t+1, \\
    c(e_i)=t,c(e'_i)=t+k-1, \hspace{11mm} otherwise.
    \end{cases}
    \end{align*}

    $c$ colors the graph with $2(k-3) + 2 = 2(k-2)$ colors. On the other hand, every path of length $k$ contains either both of $e_t,e_{t+1}$  or both of $e'_t,e'_{t+1}$, hence, as such pairs have the same color, every path of length $k$ will have at most $k-1$ distinct colors. Thus $c$ defines a $P_k$-free coloring on the cycle of length $2(k-1)$.
    
    Now we prove by contradiction that $ar(G,P_k) \leq 2 (k-2)$. Assume there is a coloring with more than $2 (k-2)$ distinct colors. Hence either $P$ or $P'$ has $(k-1)$ distinct colors, let's say it is $P$, then two edges of $P'$ that are incident to end points of $P$ should be colored by one of the colors that is already in color set of $P$. So $P'$ has at most $(k-1)-2$ colors that are not in the color set of $P$, hence we have at most $(k-1) + (k-1-2) = 2 (k-2)$ distinct colors.
\end{proof}

\begin{lemma}\label[lemma]{lem:validn3coloring}
    Let $H$ be a graph isomorphic to 
        $\mathcal{P}^v$ for any $v\in V(G)$. Then there is a valid
        coloring of $H$ with $(k-2) \cdot \myconst c_k |V(G)|$ distinct colors.
\end{lemma}
\begin{proof}
Color each path of $H$ with $k-2$ colors: color two middle
edges of the path $P_i$ by color $i$, color the rest of edges by
colors $k i+j$ for
$j\in[k-3]$. By~\Cref{lem:valid C4 coloring} this is a valid
coloring. The number of distinct colors follows from the number of
paths and the number of distinct colors of each path. Such a scheme is
depicted in the lower set of paths of the~\Cref{fig:k-2k-1} (vertex $v$).
\end{proof}

The next lemma bounds the number
of distinct colors of each individual $\mathcal{P}^v$.

\begin{lemma}\label[lemma]{lem:invalid n3+ coloring}
	There is no valid coloring of $G'$ with more than $(k-2) \cdot \myconst c_k |V(G)|$ distinct colors in one $\mathcal{P}^v$ for $v\in V(G), |V(G)|\geq 2$.
\end{lemma}
\begin{proof}
    For the sake of contradiction suppose there is a valid coloring of
    $G'$ so that $\mathcal{P}^v$ is colored with more than $(k-2)
        \cdot \myconst c_k |V(G)|$ distinct colors. By Pigeonhole Principle, at
    least one of the $P^v_i$'s has $(k-1)$ edges with distinct colors $c_1$,
    $\ldots$, $c_{k-1}$. By \Cref{lem:valid C4 coloring} all other edges
    in $\mathcal{P}^v$ should be colored with at most $(k-3) \cdot (\myconst c_k |V(G)| -1)$ other colors, contradicting that $\mathcal{P}^v$ has more than $(k-2) \cdot \myconst c_k |V(G)|$ distinct colors.
\end{proof}

\begin{definition}[Family of Distinct Colored Paths]\label[definition]{def:coloredpaths}
	A set of paths $\mathcal{P}$ is a family of distinct colored paths if the
	following conditions hold:
	\begin{enumerate}
		\item Their union is a graph with a valid $P_k$-free coloring.
		\item For every $P\neq Q\in \mathcal{P}$ and, for every $e\in P,e'\in
		Q$ we have that $c(e)\neq c(e')$.
	\end{enumerate}
\end{definition}

Note that from the above \Cref{def:coloredpaths}, it is clear that the set of paths should
be pairwise edge disjoint (otherwise it does not meet the second
condition), also one path may repeat some of its own colors.

The following lemma, basically states that we cannot have two adjacent
nodes $u, v$ in $G$ such that their
corresponding paths receive many distinct colors in $G'$. 
We employ this key property later in the hardness proof to
obtain an MIS based on the size of the family of distinct colored paths.

\begin{figure}
\begin{center}
 \includegraphics[scale=0.34]{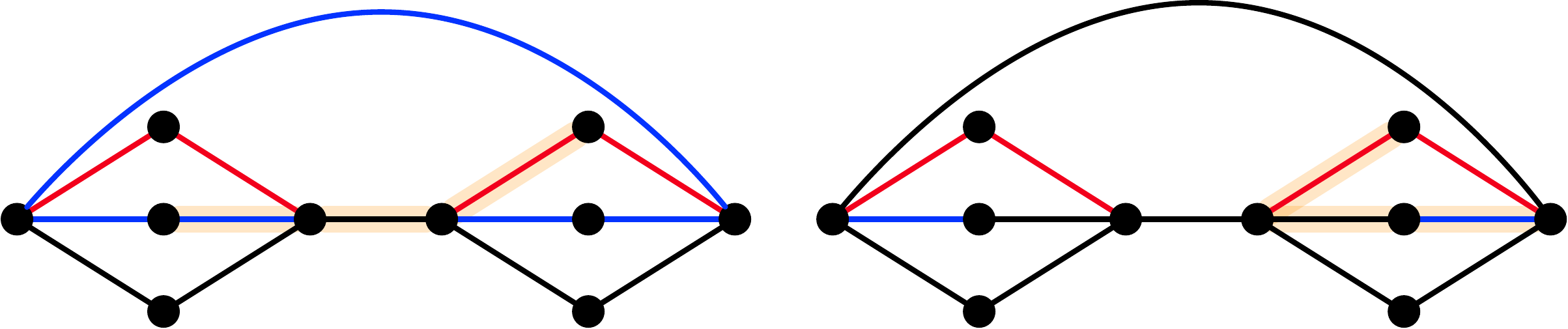}
\end{center}

    \caption{Illustration of \Cref{lem:3colors} for $P_3$. Highlighted paths
  are rainbow $P_3$'s. Neither all possible configurations are
  depicted nor all connector edges. 
No two adjacent nodes could receive three
      different colors and at the same time avoid a rainbow $P_3$.} 
      \label{fig:p3coloring}
\end{figure}

\begin{lemma}\label[lemma]{lem:3colors}
	Let $\{v, u\} \in E(G)$, then there is a constant $f_k$ (this is what we used to construct $G'$),
	depending only on $k$, such that, in any valid coloring of $G'$ if
	there are families of distinct colored paths $\mathcal{P}\subseteq
	\mathcal{P}^v, \mathcal{Q}\subseteq \mathcal{P}^u$,
	such that each $P\in\mathcal{P}\cup \mathcal{Q}$ is colored with at least
	$k-2$ distinct colors, 
	then $\min \{\sizeof{\mathcal{P}},\sizeof{\mathcal{Q}}\} < f_k$.
\end{lemma}
\begin{proof}
    We set $f_k = k+2$ and prove the lemma by contradiction. 
    Suppose $\sizeof{\mathcal{P}},\sizeof{\mathcal{Q}} \ge
        f_k$. Each path $P$ in
        $\mathcal{P}$ has at least $k-2$ distinct
        colors, if we divide it into two equal sized subpaths, one of them
        is rainbow. There are $k+2$ such subpaths in
        $\mathcal{P}$, so w.l.o.g.\ at least half of
        them are incident to $s_v$.
        Let us call them the set 
        $\mathcal{P}' = \{P_1,\ldots,P_t\}$, where
        $t=\lceil{\frac{k+2}{2}}\rceil=\frac{k-1}{2}+2$. Similarly there are $t$ rainbow subpaths $\mathcal{Q}' =
    \{Q_1,\ldots,Q_t\}$ in $\mathcal{P}^u$ such that w.l.o.g.\ they have
        $s_u$ as one of their endpoints and length of each of them is
        $\frac{k-1}{2}$.
    
    Let $c_1$ be the color of the edge $\{s_u,s_v\}$. As both of
    $\mathcal{P},\mathcal{Q}$ are families of distinct color paths, the same
    holds for $\mathcal{P}', \mathcal{Q}'$. Then, we have at least $t-1$ paths
    $\mathcal{P}''\subseteq \mathcal{P}'$ and at least $t-1$ paths
    $\mathcal{Q}''\subseteq \mathcal{Q}'$ such that none of
    their edges are colored with $c_1$. 
        Length of a path $P\in\mathcal{P}''$ is $t-2$ so it can
        have common color with at most $t-2$ paths in $\mathcal{Q}''$.
        Hence, as the number of paths in $\mathcal{Q}''$ is $t-1$, there are at least two paths $P\in \mathcal{P}'', Q\in \mathcal{Q}''$
    such that union of $P$ and $Q$ and $\{s_u,s_v\}$ is a rainbow
        path of length $k$.
A contradiction to the assumption of the lemma, hence, the claim of the lemma follows.
     
\end{proof}

For a better understanding of the above lemma
see~\Cref{fig:k-2k-1}. The following
establishes a lower bound on the number of distinct colors w.r.t.\ the size
of a maximum independent set $I$.

\begin{figure}
\begin{center}
 \includegraphics[scale=0.28]{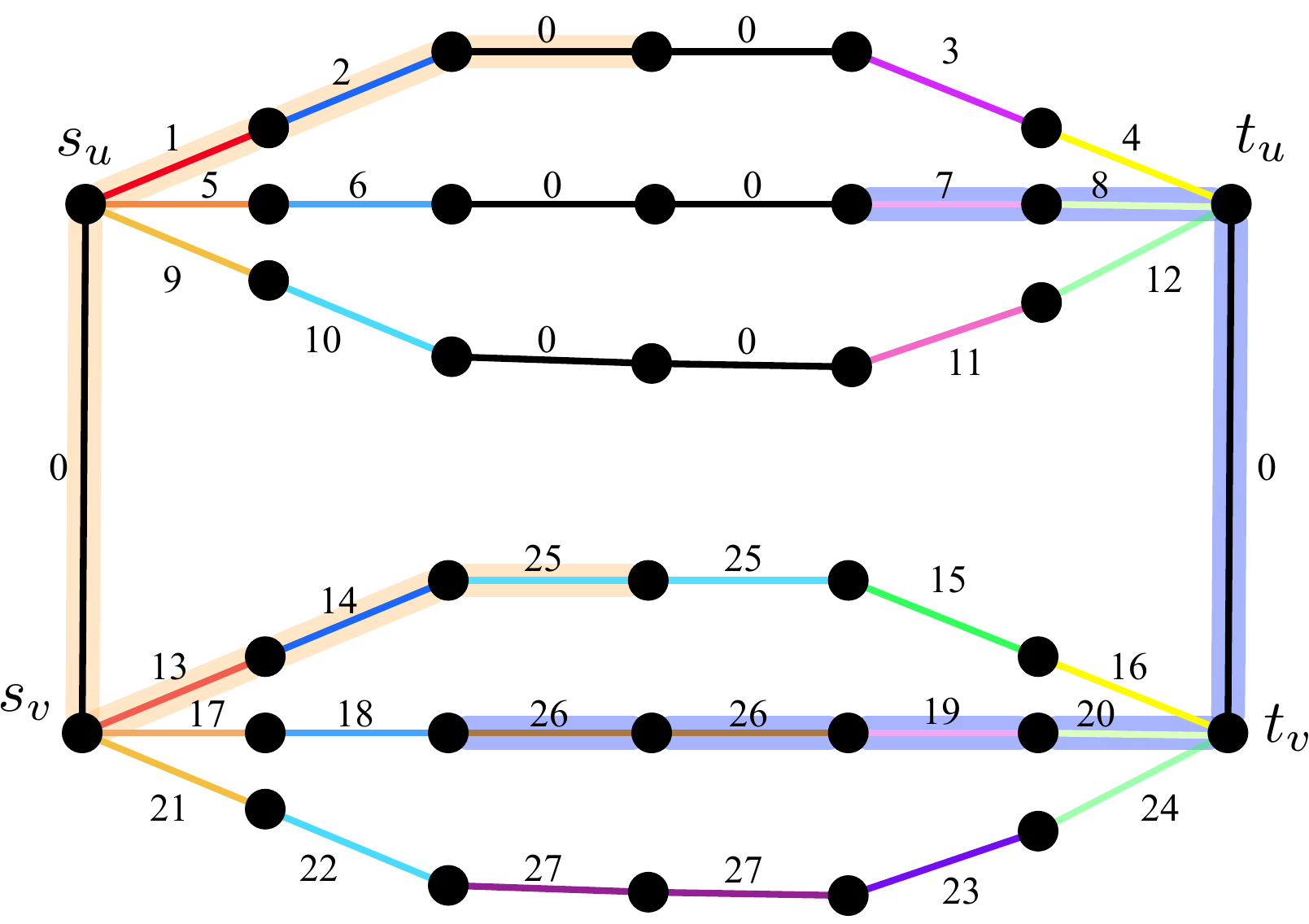}
\end{center}
    \caption{The coloring scheme of vertex gadgets for $P_7$-free
      coloring. Colors
      are represented by numbers. To simplify the visualization,
      some connector edges and some parallel paths are not drawn. For
      $v\in I$ each path gets $k-2=5$ colors and for $u\in V\setminus I$ each
      path gets $k-3=4$ colors. Two paths of length $7$ are
      highlighted, neither of them are rainbow.} 
      \label{fig:k-2k-1}
\end{figure}

\begin{lemma}\label[lemma]{lem:lowerbound}
	$ar(G',P_k) > |I| (k-2) \myconst c_k |V(G)| +
        (|V(G)|-|I|) (k-3) \myconst c_k |V(G)|$
\end{lemma}

\begin{proof}
    For $v \in I$ color $\mathcal{P}^v$ with $(k-2) \myconst
        c_k |V(G)|$ distinct colors as described in the proof of~\Cref{lem:validn3coloring} and, for $v \notin I$ color
        $\mathcal{P}^v$ with $(k-3) \myconst c_k |V(G)|$
        new distinct colors as follows: color two middle edges in each
        $P^v_i$ with a same color $c$, assign new distinct colors to
        rest of edges of $P^v_i$. At the end color all other edges of $G'$, i.e.\ all edges in $E^s_t$ with
        the same color $c$. It is easy to see that the above coloring uses $|I| (k-2) \myconst c_k |V(G)| +
        (|V(G)|-|I|) (k-3) \myconst c_k |V(G)| + 1$ distinct colors. In the rest it is enough to show that the mentioned coloring is $P_k$-free.

By~\Cref{lem:validn3coloring} we do
        not have a rainbow $P_k$ in $\mathcal{P}^v$ for any $v \in I$. Also, we do not have a rainbow $P_k$ in $\mathcal{P}^v$ for any $v \notin I$, since every $P_k$ must contain two middle edges of one of the $P^v \in \mathcal{P}^v$ which have the same color $c$. 
        Let suppose $P$ is a path of length $k$ that is not entirely
        in $\mathcal{P}^v$, for any $v\in V(G)$. Thus, $P$ has at
        least an edge of $E^s_t$. If $P$ contains more than one edge
        of $E^s_t$ then it is not a rainbow, since all edges of
        $E^s_t$ have the same color $c$. Hence, $P$ contains exactly
        one edge $e$ of $E^s_t$ connecting paths in $\mathcal{P}^u$
        and paths in $\mathcal{P}^v$. If $P$ contains two middle edges
        of $P^v \in \mathcal{P}^v$ (similarly $P^u\in\mathcal{P}^u$)
        we are done: these middle edges have the same color so $P$ is
        not a rainbow path. Otherwise, $P$ has exactly one middle edge
        $e'$ of $P^u$ and a middle edge $e''$ of $P^v$; we know that
        one of $u$ or $v$ is not in $I$, hence one of $e'$ or $e''$
        has the color $c$, the same color as $e$. Therefore, $P$ is
        not a rainbow path. See~\Cref{fig:k-2k-1} for an illustration of the explained coloring in the proof.
    
\end{proof}

Now we can prove the hardness for every odd $k > 1$. 

\begin{lemma}\label[lemma]{lem:hardness_odd}
	For every odd $k > 1$, \coloring is \nph.
\end{lemma}

\begin{proof}
    We know that the maximum independent set problem is
        \nph. We show that we can find the maximum independent set
        of $G$ if we have a maximum $P_k$-free coloring of $G'$.    

    By \Cref{lem:lowerbound}, we know that we can color the
        graph with at least 
    \label{eq:1}\begin{align*}
    A = |I| (k-2) \myconst c_k |V(G)| + (|V(G)|-|I|) \cdot (k-3) \myconst c_k |V(G)| \text{ colors.}
    \end{align*}
    
Let $c$ be a maximum $P_k$-free coloring of $G'$, by the above
argument $c$ contains at least $A$ distinct colors. Let us define a set $X= \{ v \in G \mid \mathcal{P}^v \text{ has
  more than } (k-3) \myconst c_k |V(G)|   + 2f_k  \text{ distinct colors}\}$. 
The following claim enables us to employ~\Cref{lem:3colors} and relate the size of $X$ to the size of $I$ and therefore conclude the lemma.

\begin{claim}\label{clm:a}
For a $v\in X$, $\mathcal{P}^v$ has at least $f_k$ paths such that they form a family of distinct color paths and in addition each of these paths has at least $k-2$ distinct colors.
\end{claim}
\begin{ClaimProof}\textit{of Claim \ref{clm:a}.}
Let $v\in X$, define a set $Y=\emptyset$, add paths of $\mathcal{P}^v$
to $Y$ as follows: In the iteration $i$ take a path $P_i \in
\mathcal{P}^v$ which satisfies the following two conditions: $1)$ it
has maximum number of distinct colors, $2)$ it has at least $k-2$
distinct colors w.r.t.\ the colors of paths that are already in $Y$. 

We claim the size of $Y$ is at least $f_k$. Since there
are at most $|Y| (k-1)$ distinct colors in the set of edges in
$Y$, and given~\Cref{lem:3colors} there are at most $(k-3)\sizeof{\mathcal{P}^v-Y} = (k-3)
(\myconst c_k |V(G)| - |Y|)$ distinct colors in the remaining paths of
$\mathcal{P}^v-Y$. Adding the two together we get that the number of distinct
colors in $\mathcal{P}^v$ is bounded above by $(k-3)
\cdot (\myconst c_k |V(G)| - |Y|) + |Y|\cdot (k-1) = (k-3) (\myconst
c_k |V(G)|) + 2|Y|$ which if $|Y| < f_k$ it is less than $(k-3) \myconst c_k
|V(G)|   + 2f_k + 1$, this is a contradiction to the choice of $v$ and
in general $X$, therefore $Y$ does contain at least $f_k$ paths.

In the remaining, we claim $Y$ is a family of distinct colored
paths. Given the size of $Y$, it is enough to show that no two
distinct paths in $Y$ share the same color for any pair of their edges.

For
the sake of contradiction, suppose there are two paths, let say
$P_i,P_j\in Y, i<j$, such that there are
edges $e\in P_i,e'\in P_j$ with $c(e)=c(e')$. As $i<j$ we know that
there is at most one edge in $P_j$ of color $c(e)$ otherwise $P_j$
would have at most $k-3$ distinct colors w.r.t.\ $P_i$ which violates
the our way of construction of $Y$. On the other hand $P_j$ is
rainbow: otherwise if there are two edges $e_1,e_2\in P_j$ of the same
color, then given that the color of $e'$ is already in $P_i$, $P_j$
will have at most $k-3$ distinct colors w.r.t.\ $P_i$, a contradiction
to the construction of $Y$. Hence, $P_j$ is a rainbow path of length
$k-1$.

Now we perform a case distinction on colors of
$P_i=(\{e_1,\ldots,e_{k-1})$ to show that we could not have such a
$P_k$-free coloring and at the same time have both $P_i,P_j\in Y$ to
complete the contradiction. 
\begin{enumerate} 
\item Either $c(e_1)\neq c(e)$ or $c(e_{k-1})\neq c(e)$, let suppose
  the former. Then by concatenating $e_1$ and $P_j$ we get a rainbow
  path of length $k$ a contradiction to the fact that the original
  given coloring was a $P_k$-free coloring.
\item The only remaining case is that $c(e_1)=c(e_{k-1})=c(e)=c(e')$,
  and all other edges of $P_i,P_j$ will receive distinct colors. But
  this is not possible, because at the iteration $i$ we would have
  chosen $P_j$ over $P_i$ as it has more distinct colors w.r.t.\
  existing elements of $Y$.
\end{enumerate}

Given the above case distinction, we conclude that $Y$ is a family of
distinct colored paths as claimed.
\end{ClaimProof}

By the above claim and~\Cref{lem:3colors}, we know that $X$ is an independent
        set so $|X|\le |I|$. 
    
    In the rest of the proof, we show $|X| = |I|$, by the above it is
        enough to show that $|X|\ge |I|$. Note that it is easy to find $X$ once the coloring is given, so the lemma follows. 

    To aim a contradiction assume $|X| <  |I|$. We calculate the number of distinct colors w.r.t. $c$ and prove that it is less than $A$, a contradiction to the fact that $c$ has at least $A$ distinct colors.
    
We count the maximum number of possible colors in $c$ based on type of
edges of $G'$:
\begin{enumerate}
\item Edges in $\mathcal{P}^v$ for $v\in X$:
    by~\Cref{lem:invalid n3+ coloring} we have at
    most $B=(k-2)\myconst c_k |V(G)| \cdot |X|$ distinct colors for such edges,
\item Edges in $\mathcal{P}^v$ for $v\in V(G)-X$: at most
    \begin{align*}C=((k-3)(\myconst c_k |V(G)|)+ 2f_k)  (|V(G)| - |X|)\end{align*} 
    distinct colors for them,
\item For the remaining edges, i.e. edges in $E^s_t$: by~\Cref{lem:fewcolors} we have at most $D=2 c_k |V(G)|$ distinct colors for them.
\end{enumerate}
    
    So to arrive at a contradiction, we just need to prove
        $B+C+D<A$, if we put the numbers together we will get: $|V(G)|
        - |X| \leq  c_k |V(G)|$, but, this inequality holds for $k\le
        n, c_k\ge 1$, which concludes the contradiction as it shows $B+C+D<A$.

\begin{align*}
&(k-2) \myconst c_k |V(G)| \cdot |X| +\\ 
&((k-3)  (\myconst c_k |V(G)|)+ 2f_k) \cdot (|V(G)| - |X|) + 2 c_k |V(G)|\\
&< (k-2) \myconst c_k |V(G)| \cdot |I| + 
((k-3) (\myconst c_k |V(G)|)) \cdot (|V(G)| - |I|) +1\\
&\Leftrightarrow 0<(k-2) \myconst c_k |V(G)|\cdot(|I|-|X|)\\
&-(k-3) \myconst c_k |V(G)|(|I|-|X|)\\
&-2f_k(k-3)(|V(G))|-|X|)-2c_k |V(G)|+1\\
&\xLeftarrow{|I|-|X|\ge 1} 0<\myconst c_k |V(G)|-2f_k(k-3)(|V(G)| - |X|)-2c_k|V(G)|+1\\
&\xLeftarrow{|V(G)| - |X| \le |V(G)|} 0<\myconst c_k |V(G)|-2f_k(k-3)|V(G)|-2c_k|V(G)|+1\\
&\Leftarrow 0 < (f_k - 1)c_k|V(G)| - 2f_k(k-3)|V(G)| + 1\\
&\xLeftarrow{f_k = k+2, c_k \ge 3k} 0 < (3(k+1)k - 2(k+2)(k-3)) |V(G)| + 1 \\
&\xLeftarrow{K\ge 3} 0<(k^2+5k+12) |V(G)| + 1
\end{align*}

\end{proof}

\subsection*{Hardness of the Problem for Even Values of
  $k>4$}\label{apx:even}
\noindent 
\textbf{Assumption:} In this part we assume $k = 2t, t>2$.
\begin{definition}[\sd{d}]\label[definition]{def:twostar}
For an integer $d \geq 1$, let \sd{d} be a subdivided star, i.e.,
\sd{d} is obtained by subdividing every edge of $K_{1, d}$. We call
the corresponding vertex of $K_{1, d}$ in the partition with size one, as the
center of \sd{d}. Every subdivided edge of $K_{1, d}$ is a
\textit{branch}. Therefore, \sd{d} has exactly $d$ branches.
\end{definition}

\begin{definition}[\wasted]\label[definition]{def:wasted}
    In a coloring of $G$, we choose one arbitrary edge from each color
    and call each unchosen edge of $G$ a \wasted. 
\end{definition}

Therefore, if $D$ is a set of all \wasteds of a maximum $H$-free coloring of $G$, then $|D| + ar(G, H) = |E(G)|$.
    
\begin{definition}[$D_{l, w}$]\label[definition]{def:pc} We construct an edge gadget $D_{l,w}$ as follows.
    Let $u_1, u_2, ..., u_{l+1}$ be $l + 1$ distinct vertices. Then
    for every $i \in [l]$, we connect $u_i$ to $u_{i+1}$ by $w$
    internally disjoint paths each of length two. 
\end{definition}

We call $u_1$ \textit{head} and $u_{l+1}$ \textit{tail} of $D_{l, w}$.

\textbf{Graph Construction}
Given a graph $G$, we construct a graph $G'$ as follows.

\begin{enumerate}
\item For each vertex $v \in V(G)$ with degree $d_v$, we add one
  \sd{d_v}, named $S_v$, to $G'$. Each branch of $S_v$ corresponds to
  one of the incident edges of $v$.
\item For every edge $e = \{u, v\} \in E(G)$, we add a $D_{t - 2,
    4|E(G)| + 8}$ to $G'$, named $D_e$,  such that its head is the
  leaf of the corresponding branch of $e$ in $S_u$ and its tail is the
  leaf of the corresponding branch of $e$ in $S_v$.
\end{enumerate}

For a better understanding of the graph construction see~\Cref{fig:p6coloring}.

\begin{figure}
\begin{center}
 \includegraphics[scale=0.28]{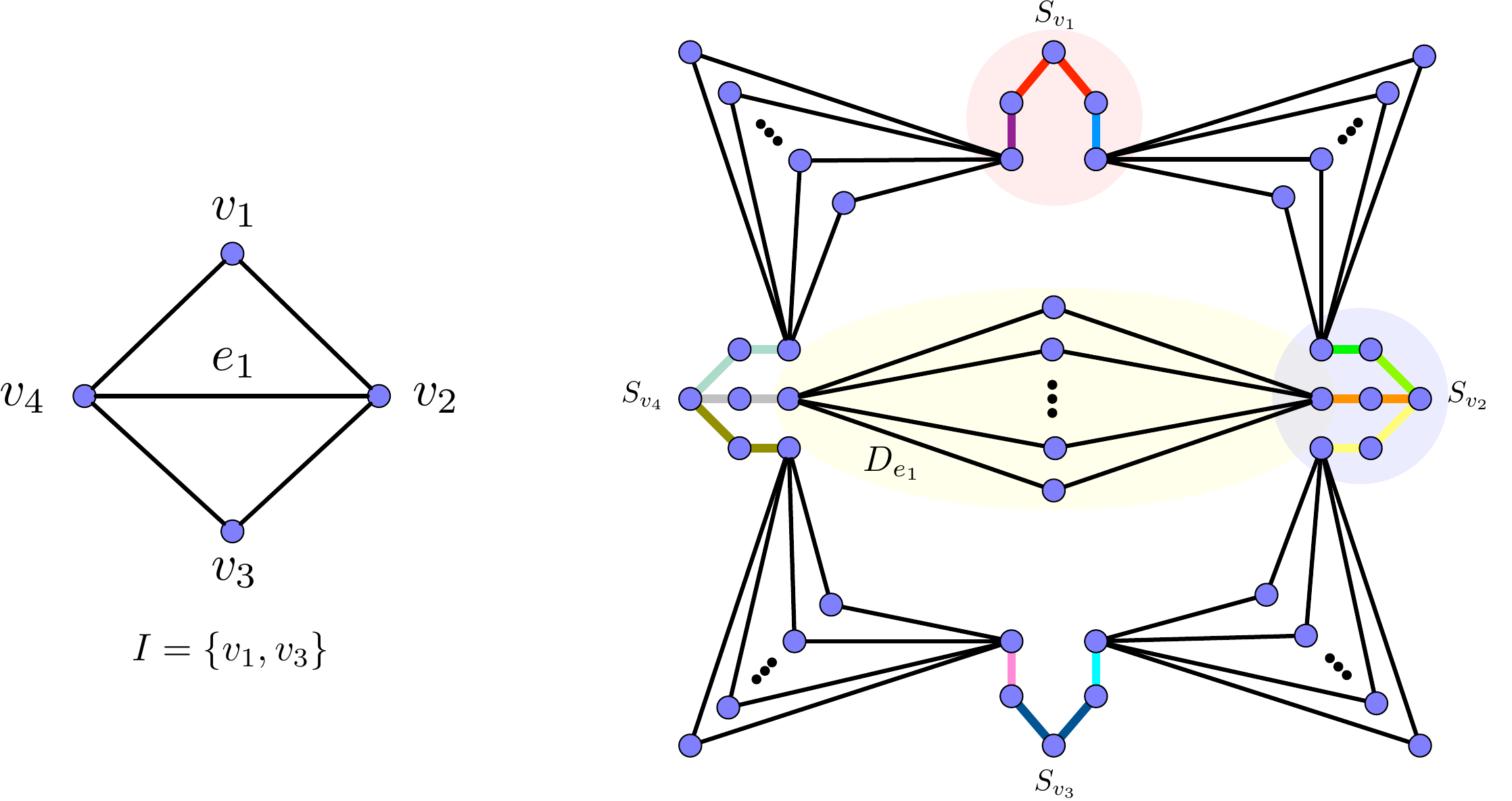}
\end{center}
    \caption{Illustration of graph construction for $k = 6$. The left figure shows the graph $G$, and the right figure shows its corresponding $G'$. All black edges of $G'$ have some unique new color. The coloring is the maximum $P_6$-free coloring of $G'$.}
      \label{fig:p6coloring}
\end{figure}

\begin{lemma}\label[lemma]{lem:eightcolor}
    In any maximum $P_k$-free coloring $c$ of $G'$, for every $D_e$, $e
    \in E(G)$, there exist at least eight edge disjoint paths, each
    of length $2t - 4$ between its head and tail such that their union
    is rainbow.
\end{lemma}

\begin{proof}
    We prove the lemma by contradiction. Suppose that there exists $e
    \in E(G)$ such that $D_e$ does not satisfy the above condition.
    There exists a set of edge disjoint paths $\mathcal{Q}$ of size
    $4|E(G)| + 8$ in $D_e$, such that each of its paths has length $2t
    - 4$ and they all start from the head of $D_e$ and end in the
    tail of $D_e$. Let $\mathcal{P}$ be the maximum size
    subset of $\mathcal{Q}$ such that union of its paths is a
    rainbow. By the assumption, we have $|\mathcal{P}| < 8$. Note that
    for every $P
    \in \mathcal{Q}\setminus \mathcal{P}$, we know that either $P$ is not a
    rainbow or $c(E(P)) \cap c(E(\mathcal{P})) \neq
    \emptyset$. Therefore, there are at least $4|E(G)| + 1$ \wasteds
    in $D_e$ w.r.t.\ $c$.
    
    To arrive at a contradiction, we construct a coloring $c'$, that
    has more distinct colors than $c$. In $c'$ we color all
    edges of $S_v$ for each $v \in V(G)$ with the same color $c_1$, and color remaining edges $e'$, i.e. $e'\in E(G')\setminus E(\bigcup_{v\in V(G)}S_v)$,
    with a new color $c_{e'}$. $c'$ has $4|E(G)| - 1$ \wasteds which
    is less than the number of \wasteds in $D_e$ w.r.t. the coloring $c$.

It is enough to show that $c'$ is a valid coloring. Every path that is
entirely in $D_e$ has length at most $2t - 2$, therefore every path of
length $2t$ has at least two edges in $\bigcup_{v \in V(G)}E(S_v)$,
hence $c'$ is a $P_k$-free coloring and has more distinct colors than
$c$, a contradiction.
\end{proof}

\begin{lemma}\label[lemma]{lem:dcoloring}
	In any maximum $P_k$-free coloring of $G'$, in each $S_v$ for
        $v \in V(G)$, there are at least $d_v - 1$
        \wasteds. 
\end{lemma}
\begin{proof}
If $d_v = 1$ the lemma is obvious. We prove the lemma for $d_v > 1$, 
by contradiction. Suppose we can color $S_v$ with at most
$d_v - 2$ \wasteds. Hence, we have at least two branches $b_1, b_2$ of
$S_v$ such that $b_1 \cup b_2$ is rainbow. Let $e \in E(G)$ be the edge
that is corresponding to branch $b_1$.  By the~\Cref{lem:eightcolor},
there are eight rainbow paths in $D_e$ each of length $2t - 4$. 
Since $|E(b_1) \cup E(b_2)| = 4$, there exists a rainbow path $Q$ in
$D_e$ such that it has no common color with $E(b_1) \cup
E(b_2)$. Therefore, concatenation of $b_2$, $b_1$ and $Q$ creates a
rainbow path of length $2t$, a contradiction. So we need at least $d_v
- 1$ \wasteds.
\end{proof}

\begin{lemma}\label[lemma]{lem:dcoloring2}
 In any maximum $P_k$-free coloring of
        $G'$, for any $v \in V(G)$ if $S_v$ has $d_v - 1$ \wasteds,
        then its coloring has the following properties:
        $1)$ all incident edges of the center vertex of $S_v$ have the same
        color and $2)$ each remaining edge of $S_v$ has a distinct color.
\end{lemma}
\begin{proof}
First we prove the following claim.
\begin{cclaim}\label{clm:b}
In any maximum $P_k$-free coloring of $G'$, for any $e\in E(G)$, there exist
four edge disjoint rainbow paths each of length
$2t - 3$ in $D_e$ such that their union is rainbow and they all start
from the head of $D_e$. Similarly, there exist
four rainbow paths each of length $2t - 3$ in
$D_e$ such that their union is rainbow and they all start from the tail of $D_e$.
\end{cclaim}
\begin{ClaimProof} \textit{of Claim \ref{clm:b}.}
By~\Cref{lem:eightcolor}, 
there exists eight disjoint paths $P_1, P_2,
\ldots, P_8$, connecting the head of $D_e$ to its tail and $P_1 \cup \cdots \cup P_8$ is rainbow. For $i \in \{1, 3, 5, 7\}$, by concatenating $P_i$ with the
starting edge of $P_{i+1}$, we get a rainbow path of length $2t - 3$
which starts from the tail of $D_e$. A similar argument holds for the
head of $D_e$ and the claim follows.
\end{ClaimProof}

Let $w$ be the center vertex of
$S_v$. We prove that in any valid coloring $c$ of $S_v$ with $d_v - 1$ \wasteds
all edges incident to $w$ have the same color.
Otherwise, there are at least two
edges $e_1,e_2$ incident to $w$ such that $c(e_1)\neq c(e_2)$.
Let suppose $e_1,e'_1$ belong to the same branch $b_e$ of $S_v$. If
$c(e'_1)\notin\{c(e_1),c(e_2)\}$ then $e_2,e_1,e'_1$ and one of the $4$ rainbow paths in
$D_e$ from the above claim, form a rainbow path of length $k$, a
contradiction. Hence, $c(e'_1)$ is used in the coloring of
incident edges of $w$. With a similar argument, every edge in $S_v$
which is not incident to $w$, has the same color as one of the incident
edges of $w$. Therefore, the total number of distinct colors in $S_v$
is at most $d_v$, a contradiction that there are at most $d_v-1$
wasted edges in $S_v$. Hence, incident edges of $w$ have
the same color; then every remaining edge must have a new
distinct color otherwise there would be more than $d_v - 1$ \wasteds in $S_v$, thus
the lemma follows.
\end{proof}

\begin{lemma}\label[lemma]{lem:adjacentcoloring}
	Let $u,v \in V(G)$ and $e = \{u, v\} \in E(G)$. In any maximum
	$P_k$-free coloring of $G'$, $S_v$ has at least $d_v$ \wasteds or
	$S_u$ has at least $d_u$ \wasteds.
\end{lemma}
\begin{proof}
Suppose that $S_u$ has $d_u - 1$ \wasteds and $S_v$ has $d_v - 1$
\wasteds. By~\Cref{lem:dcoloring}, we know that the edges in $S_u\cup S_v$ that are not connected to
the center vertices of these two subdivided stars will receive
distinct colors and those that are incident to the center vertices
will receive new colors $c_v$ (edges incident to the center of $S_v$)
and $c_u$ respectively. Then by~\Cref{lem:eightcolor}, we know that
there is a rainbow path of length $2t - 4$, $P$, between the head and
the tail of $D_e$. Let
$b_u, b_v$ be the corresponding branches of $u, v$ w.r.t.\ $e$. By
concatenation of $b_u$, $P$, and $b_v$, we get a rainbow path of
length $k=2t$, a contradiction to the assumption of the lemma.
\end{proof}

\begin{lemma}\label[lemma]{lem:relationisar}
    Let $I$ be a maximum independent set of $G$ and let $D$ be the
    set of all \wasteds in a maximum $P_k$-free coloring of $G'$, then
    $|I| = 2|E(G)| - |D|$.
\end{lemma}
\begin{proof}
We provide a coloring $c$ as follows. For every $v \in I$, color $S_v$
with $d_v - 1$ \wasteds as explained in the~\Cref{lem:dcoloring2}. For
every $u \in V(G) \setminus I$, for each branch $b$ of $S_u$, we color
both of its edges with a new color, $c_{v_{b_e}}$. For every $e \in
E(G)$, we color $D_e$ as a rainbow with new distinct colors. See~\Cref{fig:p6coloring} for a better understanding of the coloring $c$.

First, we claim that $c$ is a maximum $P_k$-free coloring of $G'$ and then
we show that  $|I|$ can be derived from the size of $c$, or equivalently from
$ar(G', P_k)$.

To show that $c$ is a $P_k$-free coloring we perform a case
distinction for every path of length $k$ in $G'$, in the following
$u,v$ are two arbitrary adjacent vertices in the graph $G$:
\begin{enumerate}
\item A path $P$ between the center of $S_u$ to the center of $S_v$
  for $\{u, v\} \in E(G)$.
\item A path $P$ that contains center of $S_v$ as one of its non-leaf vertices. 
\end{enumerate}
For the first case, as $e=\{u,v\}$ by~\Cref{lem:adjacentcoloring}
w.l.o.g.\ we can suppose $S_u$ has been
colored with at least $d_u$ \wasteds. Therefore, the first two edges
of $P$ starting from the center of $S_u$ belong to a
branch $b$ of $S_u$, have the same color
$c_{u_b}$ in $c$, so $P$ is not a rainbow path.

For the second case, the path $P$ has at least one branch, $b$, of
$S_v$ and at least one incident edge to the center of $S_v$ in another
branch $b'$ of $S_v$. Hence, if we colored $S_v$ with $d_v - 1$
\wasteds, then by~\Cref{lem:dcoloring2} two edges of $P$ that are
incident to the center of $S_v$ have the same color. Otherwise, if
$S_v$ is colored with $d_v$ \wasteds, both edges of $b$ have the  same
color $c_{v_b}$, therefore $P$ is not a rainbow path.

Now we show that $c$ is a maximum $P_k$-free coloring of
$G'$. Note that by~\Cref{lem:dcoloring}, the minimum number of wasted edges in an individual $S_v$ for $v \in V(G)$ is at least $d_v - 1$. Observe that by~\Cref{lem:adjacentcoloring}, number of $S_v$'s for
$v \in V(G)$ with $d_v - 1$ \wasteds is at most $|I|$. Moreover, in $c$, number of such $S_v$'s is exactly $|I|$ which is the maximum possible number of them. Also, for each remaining vertex, $v \in V(G)$ , $S_v$ has exactly $d_v$ \wasteds (the minimum number of possible wasted edges other than $d_v - 1$). Also, $c$ does not have any \wasted in the rest of $G'$. Therefore, $c$ has the least number of \wasteds. Hence, $c$
has the maximum number of distinct colors in any $P_k$-free coloring
of $G'$.

Total number of \wasteds in $c$ is $|D| = \sum_{v \in I}
(d_v - 1) + \sum_{v \notin I} d_v$. Hence, we get that $|I| = 2|E(G)| - |D|$ as claimed.
\end{proof}

Hence, we get the following.
\begin{lemma}\label[lemma]{lem:hardness_even}
    For every even $k > 4$, \coloring is \nph.
\end{lemma}
\begin{proof}
	By~\Cref{lem:relationisar}, we know that solving the maximum $P_k$-free coloring of $G'$ results in the size of the maximum independent set of $G$ which is \nph.
\end{proof}

\subsection*{Hardness of the Problem for $k=4$}

\textbf{Graph Construction}
Given a graph $G$, we construct a graph $G'$ as follows.

\begin{enumerate}
\item For each vertex $v \in V(G)$ with degree $d_v$, we add one \sd{d_v}, named $S_v$, to $G'$. Each branch of $S_v$ corresponds to one of the incident edges of $v$.
\item For each edge $e = \{u, v\} \in E(G)$, we merge the leaf of the
  corresponding branch of $e$ in $S_u$ with the leaf of the
  corresponding branch of $e$ in $S_v$ and call the merged vertex
  $v_e$. In addition, we add $4|E(G)| + 4$ new vertices and connect them to $v_e$. We
  call the set of edges between $v_e$ and them $L_e$.
\end{enumerate}

\begin{lemma}\label[lemma]{lem:fourcolor}
    In any maximum $P_4$-free coloring of $G'$, for each $e \in E(G)$, there exist at least four edges in $L_e$ such that their union is rainbow.
\end{lemma}

\begin{proof}
    We prove the lemma by contradiction. Suppose that in a maximum $P_4$-free coloring $c$ of $G'$ there exists $e \in E(G)$ such that $L_e$ does not contain four edges for which their union is a rainbow.
    
    Let $F \subseteq L_e$ be a maximum size subset of $L_e$ such that union of its edges is rainbow. By the assumption we have $|F| < 4$. Therefore, there are at least $4|E(G)| + 1$ \wasteds in $L_e$. 
    For the sake of a contradiction, consider a coloring of $G'$ such that each edge of $S_v$ for every $v \in V(G)$ have the same color and all other edges of $G'$ have a distinct color. The proposed coloring has exactly $4|E(G)| - 1$ \wasteds which is less than the number of \wasteds in $L_e$. Inside each $L_e$, the length of the longest path is at most $2$. Hence our coloring is a $P_4$-free coloring and has more distinct colors than $c$, a contradiction.   
\end{proof}

\begin{lemma}\label[lemma]{lem:dcoloring1}
    In any maximum $P_4$-free coloring of $G'$, in every $S_v$ for
        $v \in V(G)$, there are at least $d_v - 1$
        \wasteds.
\end{lemma}
\begin{proof}
If $d_v = 1$ the lemma is obvious, therefore we prove the lemma for $d_v > 1$. For the sake of contradiction, suppose we can color $S_v$
with at most $d_v - 2$ \wasteds. Hence, we have at least two branches
$b_1, b_2$ of $S_v$ such that union of their edges is a rainbow. Therefore, concatenation of $b_1$ and $b_2$ creates a rainbow
path of length $4$, a contradiction. Hence, we need at least $d_v - 1$
\wasteds.
\end{proof}

\begin{lemma}\label[lemma]{lem:dcoloring4}
 In any maximum $P_4$-free coloring of
        $G'$, for any $v \in V(G)$ if $S_v$ has exactly $d_v - 1$ \wasteds,
        then its coloring has the following properties:
        $1)$ all incident edges of the center vertex of $S_v$ have the same
        color and $2)$ each remaining edge of $S_v$ has a distinct color.
\end{lemma}
\begin{proof}
We prove that the coloring of $S_v$ with $d_v - 1$ \wasteds should
have the two properties mentioned in the statement of the lemma.

Let $w$ be the center vertex of
$S_v$. For the sake of contradiction, suppose there are at least two
edges $e_1,e_2$, such that they have distinct colors and they are incident to $w$. 
Let suppose $e_1$ belongs to a branch $b_e$ of $S_v$. Let the other edge
of $b_e$ be $e'_1$. If $e'_1$ does not have the same color as either of $e_1$ or
$e_2$, then $e_2,e_1,e'_1$ and one of the $4$ rainbow edges in
$L_e$ from the \Cref{lem:fourcolor}, form a rainbow path of length $4$, a
contradiction to assumption that our coloring is a $P_4$-free coloring. Hence, $c(e'_1)$ is also used in the incident edges of $w$. With a similar approach, each edge in $S_v$ which is not incident to $w$, has a common color to one of incident edge of $w$. Therefore, the total number of distinct colors in $S_v$ is at most $d_v$. Hence, there are at least $d_v$ \wasteds in $S_v$, a contradiction.
\end{proof}

\begin{lemma}\label[lemma]{lem:adjacentcoloring1}
    Let $u,v \in V(G)$ and $e = \{u, v\} \in E(G)$. In any maximum
    $P_4$-free coloring of $G'$, $S_v$ has at least $d_v$ \wasteds or
    $S_u$ has at least $d_u$ \wasteds.
\end{lemma}
\begin{proof}
Suppose that $S_u$ has $d_u - 1$ \wasteds and $S_v$ has $d_v - 1$
\wasteds. By  \Cref{lem:dcoloring4}, we know that the coloring of $S_v$ and $S_u$ must have the mentioned properties in the \Cref{lem:dcoloring4} statement. Let $b_u, b_v$ be the branches of $u, v$ that corresponds to $e$. Note that $c(E(b_u)) \cap c(E(b_v)) = \emptyset$, otherwise $S_v$ has at least $d_v$ \wasteds or $S_u$ has at least $d_u$ \wasteds, a contradiction. Hence, by concatenation of $b_u$,
and $b_v$, we get a rainbow path of length $4$, a contradiction.
\end{proof}

\begin{lemma}\label[lemma]{lem:relationisar_four}
    Let $I$ be a maximum independent set of $G$ and $D$ the set of all \wasteds in a maximum $P_4$-free coloring of $G'$, then $|I| =
    2|E(G)| - |D|$.
\end{lemma}
\begin{proof}
    
For every $v \in I$, we color $S_v$ with $d_v - 1$ \wasteds as explained in the
\Cref{lem:dcoloring4}. For every $u \in V(G) \setminus I$, for each
branch $b$ of $S_u$, we color both of its edges with a new color,
$c_{v_b}$. For every $e \in E(G)$, we color $L_e$ as a rainbow with
new colors.

We claim that the above coloring is a maximum $P_4$-free coloring of $G'$
and show that  $|I|$ can be derived from $ar(G', P_4)$.

First, we prove it is a $P_4$-free coloring. 
There are two cases for any $P_4$ in $G'$:

\begin{enumerate}
\item A path between the center of $S_u$ to the center of $S_v$ for $\{u, v\} \in E(G)$.
\item A path that contains the center of $S_v$ as one of its non-leaf
  vertices.
\end{enumerate}

For the first case, at least one of $u$ and $v$ are not in $I$ for
$\{u, v\} \in E(G)$. Assume, w.l.o.g. $u \notin I$, hence $S_u$ has
been colored with $d_u$ \wasteds. Therefore, the two first edges of
this $P_k$ starting from the center of $S_u$ are a branch, $b$, in
$S_u$, then these two edges have same color $c_{u_b}$, so this $P_4$
is not a rainbow.

For the second case, the path has at least one branch $b$ of $S_v$
and at least one incident edge to the center of $S_v$ in another
branch of $S_v$, since $k = 4$. Hence, if we colored $S_v$ with $d_v -
1$ \wasteds, then two edges of the path that are incident to the center
of $S_v$ have the same color. Otherwise, we colored $S_v$ with $d_v$
\wasteds and both edges of $b$ have same color $c_{v_b}$, so the path
is not rainbow.

Now we prove that the mentioned coloring is a maximum $P_4$-free
coloring of $G'$. Note that by \Cref{lem:adjacentcoloring1}, the number of
$S_v$'s for $v \in V(G)$ with $d_v - 1$ \wasteds is at most $|I|$ and
in our coloring it is exactly $|I|$ and all others have at least $d_v$
\wasteds and in our coloring they have exactly $d_v$ \wasteds. Also,
we do not have any \wasted in the rest of $G'$. So our coloring has
the least number of \wasteds. Hence, our coloring has the most number
of distinct colors.

The total number of \wasteds in our coloring is:
\begin{align*}
    |D| = \sum_{v \in I} d_v - 1 + \sum_{v \notin I} d_v = \sum_{v \in
      V(G)} d_v - |I| = 2|E(G)| - |I| 
\end{align*} 
Hence, $|I| = 2|E(G)| - |D|$ as claimed.
\end{proof}

\begin{lemma}\label[lemma]{lem:hardness_four}
    For $k = 4$, \coloring is \nph.
\end{lemma}

\begin{proof}
We know that finding the size of the maximum independent set is \nph
and by \Cref{lem:relationisar_four} we know that solving the maximum
$P_4$-free coloring of $G'$ results in the size of the maximum
independent set of $G$.
\end{proof}

\begin{proof}[Proof of \Cref{thm:hardness}]
    By  \Cref{lem:hardness_even}, \Cref{lem:hardness_four}, and
    \Cref{lem:hardness_odd} we show that for every integer $k > 2$ the
    problem is hard.
\end{proof}

\section{Inapproximability of $P_3$ Anti-Ramsey Coloring}
\label{sec:inapproximability_p3}

In this section, we show that for every $\varepsilon > 0$ there is no
polynomial time $\frac{1}{\sqrt{|V(G)|}^{1-\varepsilon}}$-approximation for $P_3$-free
coloring unless $P{}={}NP$~\cite{clique}, 
or similarly there is no polynomial $\frac{1}{\sqrt[4]{|V(G)|}^{1-\epsilon}}$-approximation
to estimate $ar(G,P_3)$ unless $P{}={}NP$. 
We use basic building blocks from the previous sections and prove the hardness
via a gap preserving reduction from the maximum independent set
problem.

\begin{lemma}\label[lemma]{lem:fewdistinctcolors}
In any $P_3$-free coloring of $G$ there are at most $|V(G)|$ distinct colors.
\end{lemma}

\begin{proof}
    Let $c$ be a $P_3$-free coloring of $G$ with maximum number of distinct colors and let $G'\subseteq G$ be an edge minimal subgraph of $G$ which is colored by $ar(G,P_3)$ distinct colors w.r.t. $c$ and let $G_1, G_2 \dots, G_k\subseteq G'$ be the components of $G'$. Each $G_i$, $i\in [k]$, is rainbow colored otherwise it contradicts to the edge minimality condition of our choice of $G'$. 
    
    We prove that the number of edges in each $G_i$ is at most $|V(G_i)|$ and thus,  the lemma follows. In particular, we prove that for all $t \in [k]$, it holds that $G_t$ is either a star or a triangle.
    
    Fix $t$ and let $v$ be a vertex of maximum degree in $G_t$. Let $ v_1, \ldots,
    v_{\sizeof{N(v)}} $ be neighbors of $v$.
    If $\sizeof{N(v)}= 1$ then $G_t$ is a star. 
    
    If $\sizeof{N(v)}= 2$, then $G_t$ is a star, otherwise there exists
    an edge $e = \{v_1, u\}$ or  $e = \{v_2, u\}$. Assume w.l.o.g. that $e = \{v_1, u\}$.  If $u = v_2$, it's a triangle. Otherwise we have path of length $3$: $(v_2, v,v_1, u)$.
    
    If $\sizeof{N(v)} \geq 3$,  then $G_t$ is star. Otherwise, there are two possibilities: a)  there is an edge $e
    = \{ v_i, v_j \}$ ($i,j \in [\sizeof{N(v)}] , i \neq j$) and we have $P_3$ $(v_j, v_i, v, v_z) $ (for $z \in  [\sizeof{N(v)}] ,  z
    \neq j,i$);  or b)  there is an edge $e' = \{ v_i, u \}$ ($ i \in [\sizeof{N(v)}]$)
    and we have a $P_3$ $(u, v_i,  v, v_z) $ (for $z
    \in  [\sizeof{N(v)}] ,  z  \neq i$).
    
    If $G_i$ is a star then $|E(G_i)| + 1= |V(G_i)|$.
    If $G_i$ is a triangle then $|E(G_i)| = |V(G_i)| = 3$.
    
    Thus, the lemma follows.
    
\end{proof}

\textbf{Graph Construction}
Given an undirected graph $G$, we construct a $3$-partite graph $G'$ as follows:

\begin{enumerate}
    \item For each $ v \in V(G) $ we introduce two new
    vertices $s_v, t_v \in V(G')$ and $\myconsttt |V(G)|$ internally disjoint paths of length two, $\mathcal{P}^v = \{ P^v_1, \ldots,
    P^v_{\myconsttt |V(G)|} \}$, connecting $s_v$ to $t_v$.  
    
    \item For each edge $\{v, u\} \in E(G)$, add two new edges in $E(G')$:
    $\{s_v, t_u\}$, $\{t_v, s_u\}$. We call this set of edges $E^s_t$.
\end{enumerate}

Similar to the previous sections, we say that an edge coloring is
\emph{valid} if it is a $P_3$-free coloring. The following lemmas are
similar to the ones for $P_k$ when $k$ is odd, however, there are minor differences in
some cases, so we repeat some of them customized for $P_3$. 

\begin{lemma}\label[observation]{obs:invalidn3+coloring3}
    There is no valid coloring of $G'$ with more than $\myconsttt |V(G)|$ colors in $\mathcal{P}^v$ for any $v\in V(G)$.
\end{lemma}
\begin{proof}
The proof of this lemma is similar to the proof of~\Cref{lem:invalid n3+ coloring}.
\end{proof}

\begin{lemma}\label[lemma]{lem:3color}
    Let  $\{v, u\} \in E(G)$. In any $P_3$-free coloring of $G'$, if
    there are at least three distinct colors
    in $\mathcal{P}^v$ then $\mathcal{P}^u$ is colored with at most two colors.
\end{lemma}
\begin{proof}

First, we claim that if $\mathcal{P}^v$ is colored with at least three distinct colors then
$s_v$ and $t_v$ are incidents to three edges with distinct colors.
Assume the contrary,
then w.l.o.g. $s_v$ is incident to two edges $e_1 = \{s_v, w_1 \} , e_2 = \{ s_v , w_2 \}$ of
distinct colors and there is an edge $e_3=\{t_v,w_3\}$ incident to $t_v$ such that
$c(e_1)\neq c(e_2)\neq c(e_3)$. Then, it holds that $w_3 \neq w_2$ as otherwise we get a
rainbow colored path $(t_v , w_2, s_v, w_1)$, similarly $w_3 \neq w_1$. 
Consider an edge $e'=\{w_3,s_v\}$. We show that $c(e') = c(e_3)$ and thus we obtain a contradiction.

Assume that $c(e') \in \{ c(e_1), c(e_2) \}$. If $c(e')=c(e_1)$ (or $c(e')=c(e_2)$) the path
$(t_v,w_3,s_v,w_2)$ (or $(t_v,w_3,s_v,w_1)$) is a rainbow colored path,
hence $c(e') = c(e_3)$. Thus there are three edges of distinct colors
incident to $s_v$ and it follows there are at least three edges of distinct colors incident
to $t_v$. 

Now suppose $\mathcal{P}^u$ has at least three distinct colors then both of its endpoints ($s_u$ and $t_u$) are incident to three edges of distinct
colors but those edges with edge $e=\{s_u,t_v\}$ (or
$\{s_v,t_u\}$) and three edges of distinct colors incident to $s_v$ and $t_v$ will result in a
rainbow path of length three.
\end{proof}

\begin{lemma}\label[lemma]{lem:lowerbound3}
    $ar(G',P_3) > \myconsttt |V(G)|\cdot |I|$ 
\end{lemma}

\begin{proof}

    For $v \in
        I$ color $\mathcal{P}^v$ with $\myconsttt |V(G)|$ different
        colors such that the two edges of the path $P_i$ get the same color and all other edges of $G'$ with the same color
        $c_0$. 
        There is no rainbow colored $P_3$ in $\mathcal{P}^v$ for all $v\in I$ 
        and all other $P_3$'s have at least two edges with color $c_0$ 
        or there is a $P^v_i$ which they contain it.
\end{proof}

\begin{theorem}\label{thm:Inapproximation}
	Unless $P{}={}NP$, for any fixed $\delta > 0$, there is no
	polynomial time $\frac{1}{\sqrt{|V(G)|}^{1-
			\delta}}$-approximation for $P_3$-free coloring even in
	$3$-partite graphs.
\end{theorem}
\begin{proof}
First of all, note that the graph $G'$ constructed above is a
$3$-partite graph: put every $s_v$ for $v\in V(G)$ in part $1$, every
$t_v$ in part $2$ and every other vertex in part $3$. 

We provide a reduction from the independent set problem. 
More precisely we know there is no polynomial time $\frac{1}{|V(G)|^{1- \varepsilon}}$-approximation for MIS for any fixed $\varepsilon >
0$~\cite{clique} unless $P{}={}NP$. 
We show that
if there is a $\frac{3}{\sqrt{|V(G)|}^{1-\varepsilon'}}$-approximation
for $P_3$-free coloring (for any constant $\varepsilon'$) then there
is a $\frac{1}{|V(G)|^{1- \varepsilon}}$-approximation for MIS in
polynomial time. 

Assume that there is a $\frac{3}{\sqrt{|V(G')|}^{1-\varepsilon'}}$-approximation
for $P_3$-free coloring. The graph $G'$ has $\myconsttt |V(G)|^2 + 2  |V(G)|$ vertices. 
By \Cref{lem:lowerbound3}, we conclude that we have at least 
\begin{align*}
\biggl \lceil \frac{3}{\sqrt{\myconsttt |V(G)|^2 + 2  |V(G)|}^{1-\varepsilon'}} \cdot |I| \cdot \myconsttt |V(G)| \biggr \rceil \text{ colors.}
\end{align*}

Let $X= \{ v \in G \mid \mathcal{P}^v \text{ has more than 2
  colors}\}$. By \Cref{lem:3color}, we know that $X$ is an independent
set. 

Now we just need to show that $|X| \geq \bigl \lceil \frac{1}{|V(G)|^{1-
    \varepsilon}} \cdot |I| \bigr \rceil$. 

To aim contradiction assume $|X| < \bigl\lceil
\frac{1}{|V(G)|^{1-\varepsilon}} \cdot |I| \bigr\rceil$. 
We calculate the maximum number of colors and prove that it is less than a $\bigl \lceil
\frac{3}{\sqrt{\myconsttt |V(G)|^2 + 2  |V(G)|}^{1-\varepsilon'}}  \cdot |I| \cdot \myconsttt  |V(G)|
 \bigr \rceil$. 
We have at
most $\myconsttt |V(G)| \cdot |X|$ colors for $\mathcal{P}^v$s in $X$,
$2  (|V(G)| - |X|)$ colors for other $\mathcal{P}^v$s and $2
|V(G)|$ for the remaining edges by~\Cref{lem:fewdistinctcolors}.
So we have:

\begin{align*}
&\myconsttt |V(G)| \cdot |X| + 2 (|V(G)| - |X|) + 2 |V(G)| \leq  \\ 
& \myconsttt  |V(G)| \cdot (|X|+1) \leq \\
&\biggl \lceil \frac{1}{|V(G)|^{1- \varepsilon}} \cdot |I|  \cdot \myconsttt  |V(G)| \biggr \rceil \leq \\
&\biggl \lceil \frac{3}{{(2|V(G)| + 1)} ^{1- \varepsilon}} \cdot |I|  \cdot \myconsttt  |V(G)| \biggr \rceil < \\
&\biggl \lceil \frac{3}{\sqrt{\myconsttt |V(G)|^2 + 2  |V(G)|} ^{1- \varepsilon}} \cdot |I|  \cdot \myconsttt  |V(G)| \biggr \rceil.
\end{align*}

This completes the proof that there is no $\frac{3}{\sqrt{|V(G)|}^{1- \varepsilon'}}$-approximation for the problem, now if we let the $n$ be big enough we can conclude that there is no $\frac{1}{\sqrt{|V(G)|}^{1- \delta}}$-approximation as well by replacing appropriate $\varepsilon'$ with $\delta$.
\end{proof}

\section{Precoloring $ar(G,P_k)$ Has No Subexponential Algorithm for all $k > 2$}
\label{apx:finegrain}

In this section, we study the complexity of exact algorithms computing
the anti-Ramsey number $ar(G,P_k)$ where $P_k$ is a path with $k$ edges. 
We now consider a variant of the problem for the exact time
complexity of the problem.

\begin{problem}[Precolored $ar(G,H)$]
The input consists of a graph $G = (V,E)$ where $E = E_1 \cup
E_2$. The edges in $E_1$ have assigned a color while the edges in
$E_2$ are uncolored.  Color the edges in $E_2$ with as
many new colors as possible such that there is no rainbow copy of $H$ in
$G$.  
\end{problem}

For this problem, we provide a \emph{fine grained reduction}
from 3SAT to show the hardness of the problem. That is, we provide an
instance of Precolored $ar(G,P_k)$ problem (for a constant $k>2$), i.e., a graph $G$ where some of the edges are precolored, that
asymptotically has the same size as the instance of the $3SAT$ problem,
and if there is a $2^{o(|E|)}$ algorithm to compute precolored
$ar(G,P_k)$ then there is a subexponential algorithm to solve the 3SAT
problem and this is impossible unless ETH fails. The main technical contribution
of this section is the following lemma, in which we construct the
aforementioned sparse graph.

\begin{figure}
	\begin{center}
		\includegraphics[width=\textwidth]{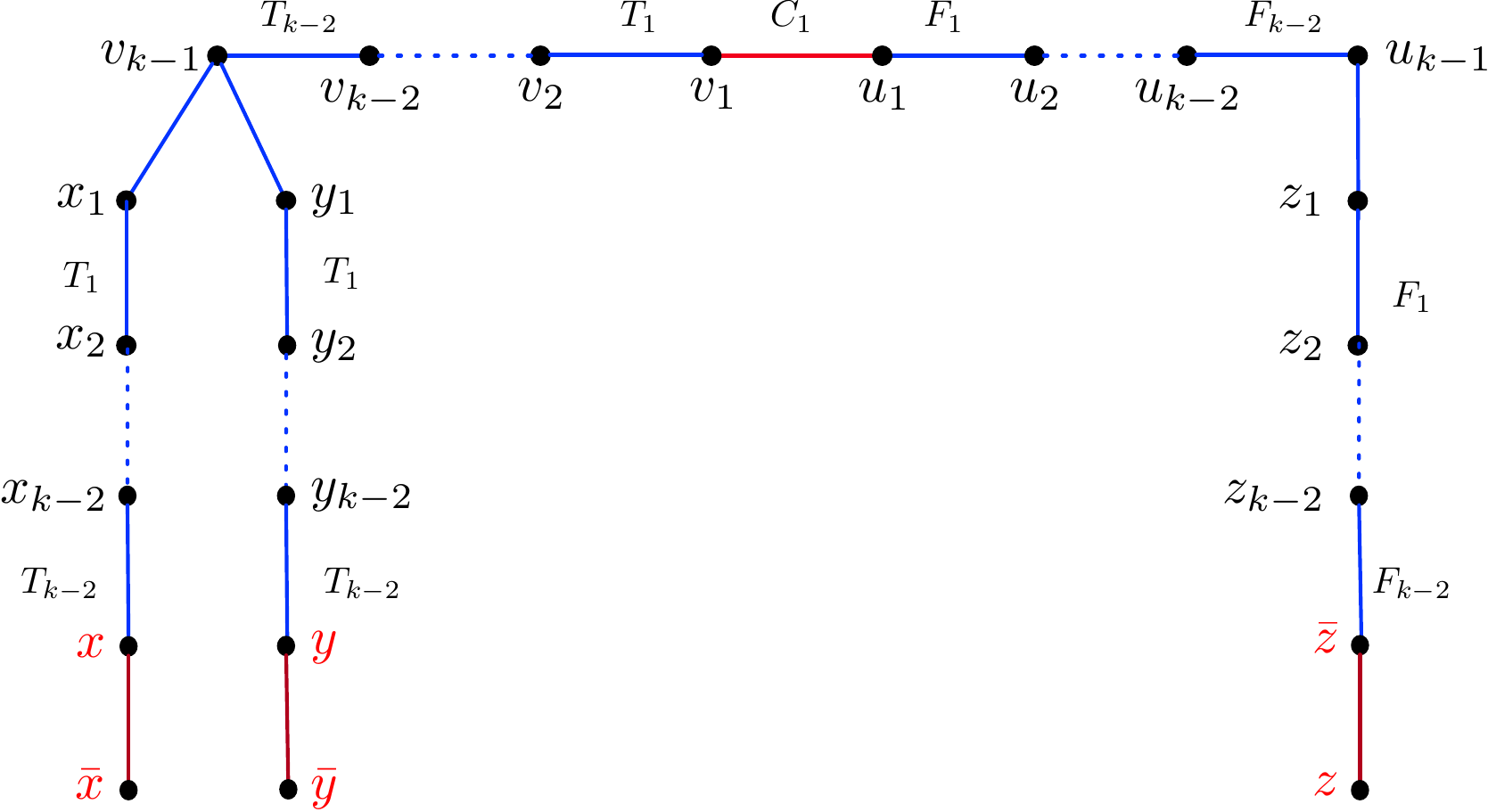}
	\end{center}
	\caption{Sketch of a gadget where the clause has exactly one negative literal. Edge $\{u_1, v_1\}$ is forced to get the color of one of its neighbors. The gadget is sparse and has $5k-7$ edges.} 
	\label{fig:new-gadget1}
\end{figure}

\begin{lemma}\label{lem:precolored_k}
The Precolored $ar(G,P_k)$ is NP-hard for fixed $k>2$.
\end{lemma}

    We show the hardness using a reduction from the 3-SAT problem. Before going to the actual proof let us go through our construction and some useful observations.
    
    \subsection*{Construction of the Graph}
     Given a Boolean formula
    $\phi$ with $n$ variables and $m$ clauses, we create a
    graph $G=(V,E)$ as follows. To simplify the understanding, with abuse of
    notation, we color some edges with colors $\{T_{1}, \ldots, T_{k-2}\}$ or $\{F_{1}, \ldots, F_{k-2}\}$---one may assume
   them to be distinct integers. 
    
        \textbf{Literal Gadgets}. For each variable $X_i \in \phi$ we create an edge $\{x_i,\bar{x}_i \}$ and we will not precolor it. We will see that its color would determine the boolean value of variable $X_i \in \phi$. In the next step, we construct the clause gadgets and connect them to the literal gadgets. 
        
      \textbf{Clause Gadgets}. For a clause $C_i = (x \vee y \vee z)$ we distinguish two cases:
        
        \begin{itemize}
            \item Either all literals are negative variables, or all are
            positive. In this case, for each literal $L \in C_i$ we create a path $(L_1, \ldots, L_{k-2}, L)$, where the last edge, i.e. $\{L_{k-2}, L\}$, is connecting our clause gadget to vertex $L$ of the literal gadget. Moreover, if $L$ is a positive literal we precolor this path by assigning color $T_j$ (or $F_j$ if $L$ is a negative literal)to the edge $\{L_j, L_{j+1} \}$ for $j<k-2$ and, precolor $\{L_{k-2}, L\}$ with $T_{k-2}$ (or $F_{k-2}$ if negative). Finally we add vertex $v$ and connect it to all 3 paths via $\{v,x_1\},\{v,y_1\},\{v,z_1\} \in E$, which are not precolored.
            
            \item Two of the literals are either positive or negative. Assume
            without loss of generality that $x$ and $y$ are both positive literals
            or both negations of variables.  Similar to the previous case, for each literal $L\in \{x,y,z\}$, we create path $(L_1, \ldots, L_{k-2}, L)$ and precolor it in the same manner. Finally instead of a vertex $v$ we create a path $(v_{k-1}, \ldots, v_1, u_1, \ldots, u_{k-1})$ in order to connect the former paths together. To do that we add edges $\{v_{k-1}, x_1\}, \{v_{k-1}, y_1\}$ and $\{u_{k-1}, z_1\}$ to $E$. Then, if $x$ and $y$ are positive literal we precolor each $\{v_j, v_{j+1} \}$ with $T_j$ ($F_j$ if $x,y$ are negative) and $\{u_j, u_{j+1} \}$ with $F_j$ ($T_j$).
            Note that we leave out $4$ edges uncolored, for now, namely $\{v_{k-1}, x_1\}, \{v_{k-1}, y_1\},\{u_{k-1}, z_1\}$ and $\{v_1, u_1\}$.
            
        \end{itemize}
    
    See \Cref{fig:new-gadget1} for a sketch of the construction of the gadgets. W.l.o.g., we assume that for every $i\in [n]$ both variable $x_i$ and its negation appear in some clauses as literals. Otherwise, if a variable appears only negative or positive, we can simply satisfy all the clauses that contain that variable. The above assumption enforces the following observation.
    
    \begin{observation}\label{obs:TF_Middle_Man}
        In a $P_k$-free coloring, for each variable $x$, $c(\{x, \bar{x}\}) \in \{T_{k-2},F_{k-2} \}$, and for each clause $c(\{v_1, u_1\}) \in \{T_{1},F_{1}\}$. 
    \end{observation}
    
    This follows from the construction, i.e., consider paths $P=(x_1, x_2, \ldots, x_{k-2},x, \bar{x}, \bar{x}_{k-2})$ and $Q=(\bar{x}_1, \ldots, \bar{x}_{k-2},\bar{x}, x, x_{k-2})$.
        $P$ contains colors $\{T_1, \ldots, T_{k-2}\} \cup \{F_{k-2} \}$ and $Q$ contains $\{F_1, \ldots, F_{k-2}\} \cup \{T_{k-2} \}$. Hence $c(\{x, \bar{x} \})$ must be either $T_{k-2}$ or $F_{k-2}$. The same argument applies to $\{v_1, u_1\}$.
        
     In the following by new color we mean a color other than $F_i$'s and $T_i$'s, i.e. preassigned colors.
    \begin{lemma}\label{lem:OneColorEachClause}
    	In a $P_k$-free coloring, there can be at most 1 new color in each clause gadget.
    \end{lemma}
    \begin{proof}
    	Towards a contradiction assume there are two new colors in clause gadget $C=z \vee y \vee z$, namely $c_1$ and $c_2$.
    	\begin{itemize}
    		\item $x, y, z$ are all positive literals or all are negative literals. There are three unassigned edges, namely $\{v, x_1\}$, $\{v, y_1\}$ and $\{v, z_1\}$. 
    		W.l.o.g. assume that color $c_1$ is assigned to $\{v, x_1\}$ and $c_2$ is assigned to $\{v, y_1\}$. Then $(y_1, v, x_1, x_2, \ldots, x_{k-2}, x)$ is a rainbow $P_k$, contradiction.
    		
    		\item $x, y$ are positive but $z$ is a negative literal (or vice-versa). If $c_1$ and $c_2$ are assigned to $\{v_{k-1},x_1\}$ and $\{v_{k-1},y_1\}$ we immediately have a rainbow $P_k$ similar to the previous case. Also remember that by~\Cref{obs:TF_Middle_Man} edge $\{u_1, v_1\}$ cannot get a new color. Thus there remain two symmetric cases: $c_1$ is assigned to $\{u_{k-1}, z_1 \}$ and $c_2$ is assigned to $\{v_{k-1}, x_1\}$ or $\{v_{k-1}, y_1 \}$. Consider the following two paths, we claim at least one of them is rainbow:
    		\begin{enumerate}
    			\item $Q_z=(z_1, u_{k-1}, u_{k-2},\ldots, u_1, v_1)$
    			\item $Q_x=(x_1, v_{k-1}, v_{k-2},\ldots, v_1, u_1)$
    		\end{enumerate}
    		Observe that $c(Q_z)= \{c_1\} \cup \{F_1, \ldots, F_{k-2} \} \cup \{c(\{v_1, u_1\}) \}$ and $c(Q_x)= \{c_2\} \cup \{T_1, \ldots, T_{k-2} \} \cup \{c(\{v_1,u_1 \}) \}$. Remember that $c(\{v_1,u_1 \}) \in \{T_1, F_1\}$. But, either of them makes at least one of $Q_x$ and $Q_z$ become rainbow, a contradiction.
    		
    	\end{itemize}
    \end{proof}

	We are ready to prove Lemma~\ref{lem:precolored_k}.
    \begin{proof}[Proof of Lemma~\ref{lem:precolored_k}]
    We claim that the formula $\phi$ is satisfiable if and only if
    $ar(G,P_k) = m + 2(k-2)$, that is there is a coloring of the edges of $G$
    with $m+2(k-2)$ colors ($\{T_1, \ldots, T_{k-2}\}$,$\{F_1, \ldots, F_{k-2}\}$ and another new $m$ colors, one for each clause).
    
    For the direct implication, if the formula $\phi$ is satisfiable, we color the edges of $G$ as follows.  For each variable $X_i$, if $X_i$
    is assigned to True, then we color the edge $\{x_i,\bar{x_i}\}$ with
    $T_{k-2}$, otherwise we color this edge with $F_{k-2}$. 
    
    Let $C_i = (x \vee y \vee z)$. Assume without loss of
    generality that $C_i$ is satisfied by the literal $z$. Then if all literals are negative or all are positive we
    color the edge $\{v,z_1\}$ with a new color. Then, if $z$ corresponds to negation of a
    variable, we color $\{v, x_1\}$ and $\{v, y_1\}$ with $F_1$ (otherwise $T_1$).
    
    Now assume that $x,y$ are positive literals but $z$ is a negative literal. Again, W.l.o.g., assume that $C_i$ is satisfied by $z$. We color $\{u_{k-1}, z_1\}$ by a new color and then if $z$ is negative, we color $\{v_1, u_1\}$ with $F_1$ ($T_1$) and $\{v_{k-1}, x_1\},\{v_{k-1}, x_1\}$ with $T_1$ ($F_1$).
	
	We now show that the coloring is valid. Suppose that there exists a rainbow $P_k$ path $Q$ that goes through clause gadget of $C_i=x \vee y \vee z$. Observe that precolored edges do not yield a rainbow path of length $k$. Thus $Q$ must contain a newly colored edge. By a case distinction on value assignment to literals of $C_i$ we have:
	\begin{enumerate}
		\item $x,y,z$ are all positive literal (or vice-versa). Assume w.l.o.g., that $C_i$ is satisfied by $z$.
		Then all of the edges of the clause gadget get the same color $T_1$ except for $\{v,z_1\}$ which has a new color.
		Note that both of the $x$-branch (the path $(v,x_1, \ldots, x_{k-2})$) and $y$-branch of the clause gadget starts with consecutive $T_1$'s, that is $\{v,x_1\}$ and $\{x_1, x_2\}$ (respectively $\{v,y_1\}$ and $\{y_1, y_2\}$) are colored with $T_1$. Hence a rainbow $P_k$ path that contains $\{v,x_1\}$ or $\{v,y_1\}$ cannot go inside $x,y$ or $z$-branch since it inevitably gets two $T_1$'s. Thus it can only contains $\{v,z_1\}$ and goes through $z$-branch, i.e. $(v, z_1, \ldots, z_{k-2},z,\bar{z})$. But the latter contains two $T_{k-2}$'s.
		
		\item  $x,y$ are positive and $z$ is a negative literal (or vice-versa). We assume that $C_i$ is satisfied by $z$, the other two cases are similar. In this case, $Q$ should contains at least one of the edges in the sets; $\{ v_1, u_1 \}$, $ \{ v_{k-1}, x_1 \} $, $\{v_{k-1}, y_1\}$, and $ \{ u_{k-1}, z_1 \}$. Otherwise due to the construction and Observation ~\ref{obs:TF_Middle_Man}, it cannot be a rainbow path of length k.
		\begin{itemize}
			\item $Q$ contains edge $\{v_1,u_1\}$. Note that $c(\{v_1,u_1\})=F_1$ therefore $Q$ cannot contain $\{u_1, u_2\}$. Thus $Q$ must be either $(u_1, v_1, \ldots, v_{k-1}, x_1)$ or $(u_1, v_1, \ldots, v_{k-1}, y_1)$, but neither of them is a rainbow path.
			\item $Q$ contains edge $\{v_{k-1}, x_1\}$ (or $\{v_{k-1}, y_1\}$). Therefore $Q$ cannot contain $\{v_{k-1}, y_1\}$, $\{x_1, x_2\}$ or $\{v_2, v_1\}$, since all of them have color $T_1$ and by Observation ~\ref{obs:TF_Middle_Man}, we know that the color of $\{v_{k-1}, x_1\}$ is also $T_1$. But if we remove these edges, the connected component which contains $\{v_{k-1}, x_1,\}$ has longest path of length $k-1$. Hence a rainbow $P_k$ is not possible.
			\item $Q$ contains edge $\{u_{k-1},z_1\}$. Then any rainbow path of length $k$ must start from some $u_s$ and end in $z_k$, i.e. $Q=(u_s, \ldots, u_{k-1}, z_1, \ldots, z_{t})$. Since $c(\{u_s, u_{s+1}\})=c(\{z_s, z_{s+1}\})$, we have that $t=s$. Thus $|Q| < k$. For the corner case of $s=k-1$ which a $\{u_{k-1}, u_k\}$ does not exist we have that $Q=(u_{k-1}, z_1, \ldots, z_{k-2},z, \bar{z})$ which is not a rainbow since $\{z, \bar{z}\}$ and $\{z_{k-1}, z\}$ are colored by $F_{k-2}$.
		\end{itemize}
		
	\end{enumerate}
	Hence the coloring is valid.

    For the reverse implication, assume that we are given a coloring of $G$ with $m+2(k-2)$ colors. We show how to recover a satisfying assignment for $\phi$. First of all, notice that in a clause gadget we can add at most one new color by Lemma~\ref{lem:OneColorEachClause}. Thus, we construct a satisfying assignment as follows. For each variable $X_i$ if the edge $\{x_i, \bar{x}_i \}$ is set to $T_{k-2}$ then we set $X_i$ to True, otherwise we set $X_i$ to False. Note that Observation~\ref{obs:TF_Middle_Man} ensures that the latter is possible.
    
    Finally, we show that this is a satisfying assignment for $\phi$. Towards a contradiction suppose that there is a clause $C_i= x\vee y\vee z$ that is not satisfied. By Lemma~\ref{lem:OneColorEachClause} each clause has exactly one new color. Thus, we have two cases:
    
    \begin{enumerate}
    	\item $x,y,z$ are all positive (or all negative). Hence $\{x,\bar{x}\},\{y,\bar{y}\}$ and $\{z,\bar{z}\}$ are all colored by $F_{k-2}$. W.l.o.g. let $\{v, x_1\}$ be the edge that gets the new color $c_1$ of clause $C_i$. Then $(v, x_1, \ldots, x_{k-1}, x , \bar{x})$ has colors $\{T_1, \ldots, T_{k-2}\} \cup \{c_1, F_{k-2}\}$. Thus the coloring is not $P_k$-free.
    	
    	\item $x,y$ are positive and $z$ is a negative literal (or vice-versa). Note that $\{v_{k-1}, x_1\}, \{v_{k-1}, y_1\}$ or $\{u_{k-1}, z_1\}$ cannot get the new color since the same argument of the above case applies. Thus it must be $\{v_1, u_1\}$ that gets the new color. But, by~\Cref{obs:TF_Middle_Man} it can only get $T_1$ or $F_1$. That contradicts with each clause getting exactly one new color.
    \end{enumerate}
Hence the assignment must satisfy $\phi$.
\end{proof}

Given the above lemma and the sparsification lemma~\cite{IMPAGLIAZZO2001512} we conclude the
following theorem.

\begin{repeathm}{thm:ethprecolored}
There is no $2^{o(|E(G)|)}$ algorithm for Precolored $ar(G,P_k)$, for any fixed $k$, unless ETH fails.
\end{repeathm}
\begin{proof}
    We may assume the 3SAT instance used in the construction of
        the graph $G$ in the proof of Lemma~\ref{lem:precolored_k} is sparse,
        that is the number of clauses $m$ is in the order of number of
        variables $n$, i.e. $m\in O(n)$. Thus by the sparsification
        lemma~\cite{IMPAGLIAZZO2001512} there is no $2^{o(n)}$ algorithm to solve 3SAT (unless ETH fails).
    
    On the other hand in the construction of the graph $G$ for each variable we have $1$ edge and for each clause, we have at most $5k-7$ edges so in total the number of edges in the graph is bounded above by $5km+n$ hence $|E(G)|\in O(n)$ for a constant $k$. Hence there is no $2^{o(|E(G)|)}$ algorithm for Precolored $ar(G,P_k)$ unless ETH fails.
\end{proof}

\section{Color Connected Coloring and its Applications}
\label{apx:approximation_trees}

In this section, we introduce the notion of color connected coloring and
using that we provide a polynomial time algorithm to compute
$ar(T,P_k)$, where $T$ is a tree. Roughly speaking, in a color connected
coloring we try to color the graph with the maximum number of colors so
that the set of edges of every color class induces a connected subgraph. The main result of this section is the following theorem.

\begin{repeathm}{thm:TreePoly}
	There is an exact algorithm that computes $ar(T,P_k)$ in linear time w.r.t. $|V(T)|$, where $T$ is a tree.
\end{repeathm}

Let $c$ be a $P_k$-free coloring of a graph $G$ and let $c_1$ be one of such
colors used in $c$. Then, we call the induced graph $G[\{v\mid \exists u\in V(G),
e=\{u,v\}\in E(G), c(e)=c_1\}]$ as an induced $c_1$-graph and we write it
$G[c_1]$. If $G[c_1]$ is connected then we say $c_1$ is a
\emph{connected color}; otherwise, it is a \emph{disjoint color}. 

\begin{definition}[Color Connected Coloring]
	Given a graph $G$, a $P_k$-free coloring $c$ of $G$ is a color connected coloring if for every color $c_i$ used in $c$, $G[c_i]$ is a connected component.
\end{definition}

In the rest of this section, we assume that $T$ is a rooted tree with $r_T$ as its root. We define $T_v$ as the largest subtree with $v \in V(T)$ as its root. Depth of a vertex $v \in V(T)$, $H_v$, is the number of edges between $v$ and the root. Furthermore, we define $C(v)$ as the set of children of $v$ in a rooted tree. As we can color the graph with at most $|E|$ many colors, in this proof we use a palette of colors $\mathcal{C}=\{c_e\mid e\in E(T)\}$. That is whenever we color an edge $e$ with a new color, its color will be $c_e$, otherwise, $e$ will get a color of one of the already colored edges.

\begin{lemma}\label[lemma]{lem:treeconnectedcoloring}
	There exists a maximum $P_k$-free coloring of $T$, which is color connected.
\end{lemma}
\begin{proof}
Let $c$ be a maximum $P_k$-free coloring of $T$ with the minimum number of color connected components. If for every $c_i$, $T[c_i]$ has one connected component we are done. Otherwise,
    towards the contradiction, let
    $c_1$ be a color used in $c$, for which $T[c_1]$ has more than one connected components, $\{T_1,\ldots,T_r\}$ for some
    $r>1$. W.l.o.g. suppose $T_1$ is the component of $T[c_1]$ with the deepest root, in other words $\argmax_{i \in [r]} \min_{u \in V(T_i)} H_u$ equals to one. Since $r > 1$, the root of subtree $T_1$, $v$, has a parent. Let $e$ be the edge between $v$ and its parent. We recolor all of $E(T_1)$ with color $c(e)$. This clearly creates a new
    coloring $c'$ with the same set of colors as $c$; however, it
    has one less color connected component than $c$ which contradicts our minimality assumption on $c$.
    To complete the contradiction, it is sufficient to show that $c'$ is a
    $P_k$-free coloring.
    
    Towards the contradiction, let $P$ be a rainbow $P_k$ in $c'$. We perform a case distinction on $|E(P) \cap E(T_1)|$ to derive a contradiction. 
    \begin{enumerate}
        \item $|E(P) \cap E(T_1)| = 0$: In this case, the coloring of $P$ in $c$ and $c'$ is identical. Moreover, $P$ is not rainbow in $c$, hence $P$ is not rainbow in $c'$ either, a contradiction.
        
        \item $|E(P) \cap E(T_1)| = 1$: In this case, let $e' \in E(P) \cap E(T_1)$ be the only edge of $P$ that is recolored in $c'$. There must exist another edge $e''$ of $P$ which is colored by $c_1$. We know that $e'' \not\in E(T_1)$, so $e''$ is not incident to $v$. We claim that $e'' \not\in E(T_v)$. Suppose by  contradiction, $e'' \in E(T_v)$. Since $e'' \not\in E(T_1)$, w.l.o.g.  assume $e'' \in E(T_2)$. Since $T_1$ and $T_2$ are two disjoint connected components in $T_v$ and $v \in V(T_1)$, $\min_{u \in V(T_1)} H_u < \min_{u \in V(T_2)} H_u$ which contradicts the fact that $T_1$ is the component of $T[c_1]$ with deepest root. We showed that $e'' \not\in E(T_v)$. Since $|E(P) \cap E(T_1)| = 1$, its obvious that $e \in E(P)$. $c'(e) = c'(e')$, a contradiction.
        
        \item $|E(P) \cap E(T_1)| > 1$: In this case, at least two edges of $P$ have the same color $c(e)$, hence $P$ is not rainbow, a contradiction.
        
    \end{enumerate}
\end{proof}

The purpose of our algorithm is to find a maximum $P_k$-free color connected coloring of a tree, $T$, since by  \Cref{lem:treeconnectedcoloring} it is a maximum $P_k$-free coloring of $T$.

\begin{definition}[$L_1^v$, $L_2^v$]\label[definition]{def:L1L2}
For a color connected coloring $c$ of $T$, we define $L_1^v$ to be a longest rainbow path in $T_v$ starting from $v$. Moreover, let $L_2^v$ be the longest rainbow path such that $L_1^v$ and $L_2^v$ are edge disjoint and $L_1^v \cup L_2^v$ is also rainbow.  
\end{definition}

\begin{lemma}\label[lemma]{lem:longestpaths}
A color connected coloring $c$ of $T$ is $P_k$-free if and only if $|E(L_1^v)|+|E(L_2^v)|<k$, for all $v \in V(T)$.
\end{lemma}
\begin{proof}
If there exist $v \in V(T)$ such that $|E(L_1^v)| + |E(L_2^v)| \geq k$, $c$ is not a $P_k$-free coloring, since $L_1^v \cup L_2^v$ is a rainbow path. 

To prove the other direction of the lemma, first we need to prove the following claim.

\begin{cclaim}\label{clm:abcd}
For any $v \in V(T)$, $L_1^v \cup L_2^v$ is a maximum length rainbow path including $v$ in $T_v$.
\end{cclaim}

\begin{ClaimProof}\textit{of Claim~\ref{clm:abcd}.}
We prove the claim by contradiction, suppose there is a rainbow path which can be partitioned as $L_3 \cup L_4$, each starting from $v$, such that $|E(L_3)|+|E(L_4)|>|E(L_1^v)|+|E(L_2^v)|$. Since $L_1^v$ is a longest rainbow path we have that $|E(L_3)|,|E(L_4)| >|E(L_2^v)|$. Hence, $L_3$ and $L_4$ must have a common color with $L_1^v$. We know that the incident edge of $v$ in each path $L_1^v, L_3, L_4$ must have the same color, since $c$ is a color connected coloring. But we assumed that $L_3 \cup L_4$ is rainbow, a contradiction. Hence, the claim is proved.
\end{ClaimProof}

Now we can prove the remaining direction of the lemma. Suppose $P$ is a rainbow path in $T_v$. Thus, $P$ can be partitioned as $P_1 \cup P_2$, each starting from $u \in V(T_v)$. Note that $|E(L_1^u)|+|E(L_2^u)| < k$ by the lemma statement. Also, by the above claim, we know $|E(P)| \leq |E(L_1^u)|+|E(L_2^u)|$. Therefore, $|E(P)| < k$ for any arbitrary rainbow path in $T_v$.
\end{proof}

\begin{definition}[$D(v,i,j)$]\label[definition]{def:dp}
Let $i \geq j$, $i + j < k$, and $v \in V(T)$, we define $D(v,i,j)$ to be the number of distinct colors in a color connected maximum $P_k$-free coloring of $T_v$ such that $|E(L_1^v)|=i$ and $|E(L_2^v)|=j$.
\end{definition}

For $e=\{u,v\}$ where $v$ is the parent of $u$, we define $T_e$ to be a subgraph of $T_v$ with $E(T_u)\cup e$ as its edge set, that is a subgraph of $T_v$ that is hanging from $e$.
    
\begin{proof}[Proof of~\Cref{thm:TreePoly}]

By~\Cref{def:dp}, we know that $ar(T,P_k)= \max \{D(r_T, i, j) | i + j < k \}$. We show that $D(v,i,j)$ can be computed using the values of $D(u, \cdot)$ for $u \in V(T_v) \setminus \{v\}$. Hence, $D(\cdot)$ can be computed by a post-order traversal of $T$.

To compute $D(v, i, j)$, if $v$ is a leaf of $T$, the only valid case is $D(v,0,0)$, since there is no edge in $T_v$. Hence, in the remaining, we suppose that $v$ is not a leaf. We proceed by case distinction based on types of children of $v$. A child $u$ of $v$ is of the following types:
\begin{inparaenum}
\item $u\in L_1^v$,
\item $u\in L_2^v$,
\item $u\notin L_1^v\cup L_2^v$
\end{inparaenum}

Now for each child $u$ of $v$ and $z \in [3]$, such that $e=\{v, u\} \in E(T)$, we define $A_{u, z}$ as the maximum number of distinct colors in $T_e$ if $u$ belongs to case $z$, such that it does not violate the definition of $D(v, i, j)$. Note that only one child of $v$ belongs to the first case. Also, for $j>0$, there is only one child of $v$ in the second case. Moreover, for $j=0$ there is not any child in the second case. All other children of $v$ belong to the third case. Therefore, we can compute $D(v, i, j)$ by \Cref{equ:jgre0} and \Cref{equ:jeq0}, for $j>0$ and $j=0$, respectively.
\setlength{\belowdisplayskip}{0pt} \setlength{\belowdisplayshortskip}{0pt}
\setlength{\abovedisplayskip}{0pt} \setlength{\abovedisplayshortskip}{0pt}
\begin{equation}\label[equation]{equ:jgre0}
D(v, i, j) = \max \{ A_{u_1, 1} + A_{u_2, 2} + \sum_{u \in C(v) \setminus \{u_1, u_2 \}}A_{u, 3} \big\vert u_1, u_2\in C(v), u_1 \neq u_2 \},
\end{equation}
\begin{equation}\label[equation]{equ:jeq0}
D(v, i, 0) = \max \{ A_{u_1, 1} + \sum_{u \in C(v) \setminus \{u_1\}} A_{u, 3} \big\vert u_1\in C(v) \}.
\end{equation} In what follows, we show how to compute the value of $A_{u, z}$.

\subparagraph*{$a)$\ $\mathbf{u\in L_1^v}$: }Let $e=\{u,v\}\in E(T)$ and $u\in L_1^v$. Then we have that $E(L_1^v) \setminus \{\{v, u\}\} $ is a rainbow path of length $i-1$. Observe that, since $c(e)$ is in at most one of $c(E(L_1^u))$ or $c(E(L_2^u))$, hence by appending $e$ to their tails, at least one of the two paths, $L_1^u$ or $L_2^u$, extends to a longer rainbow path. If $L_1^u$ extends to a longer rainbow path, we have $|E(L_1^u)| = i - 1$. Otherwise, $c(e) \in c(E(L^1_u))$ and by \Cref{def:L1L2} every rainbow path with greater length than $L^2_u$ starting from $u$ in $T_u$ has a common color with $L_1^u$. Moreover the common color is $c(e)$, since the coloring is color connected. Hence, $L_2^u$ is the longest rainbow path in $T_u$ that extends to a longer rainbow path which results in $|E(L_2^u)| = i - 1$. Therefore, $|E(L_1^u)| = i - 1$ or $|E(L_2^u)| = i - 1$. Thus, $A_{u, 1}$ equals to the maximum value obtained from these two cases.\noindent
\begin{enumerate}
\item
\textbf{$|E(L_1^u)| = i - 1$: } In this case, $e$ can get a new color $c_e$. Hence, the maximum number of distinct colors used in $T_e$ for $D(v,i,j)$ is $\max_{x<i} D(u,i-1,x) + 1$.
\item 
\textbf{$|E(L_2^u)| = i - 1, |E(L_1^u)| > i - 1:$ } Then $c(\{v, u\}) \in c(E(L_1^u))$, since the length of the longest rainbow path must not exceed $i$. Also, $e$ must have the same color as the incident edge of $u$ in $L_1^u$, since the coloring is color connected. However, in this case, $P:=L_2^u \cup e$ forms a rainbow path, since $c(e) \in c(E(L_1^u))$ and $|c(E(L_1^u)) \cap c(E(L_2^u))| = 0$. Moreover, $P$ is the longest rainbow path of $T_v$ starting with $e$, since every other path with longer length has a common color with $L_1^u$ and we are looking for a color connected coloring, thus this color is $c(e)$. So the maximum number of distinct colors used in $T_e$ for $D(v,i,j)$ in this case is $\max_{x\ge i} D(u, x, i - 1)$.
\end{enumerate}

\subparagraph*{$b)$\ $\mathbf{u\in L_2^v}$: }$A_{u, 2}$ can be computed similar to the previous case. 
\subparagraph*{$c)$\ $\mathbf{u\notin L_1^v\cup L_2^v}$: }
In the following let $e_1=\{v,u_1\}\in L_1^v$ and $e_2=\{v,u_2\}\in L_2^v$.
For every child $u$ of $v$ such that $u \notin \{u_1,u_2\}$, suppose that $x = |E(L_1^u)|, y = |E(L_2^u)|$. Also, let $e=\{u, v\}$. Hence, $A_{u, 3}$ is equal to the maximum value obtained from the following cases by iterating over all combination of $x$ and $y$ such that $x + y < k$ and $x \geq y$. 

\begin{enumerate}
    \item $x < j$: In this case, $e$ can get a new color $c_e$. Therefore, the optimal solution for this case of $T_e$ is $D(u, x, y) + 1$.

    \item $j \leq x < i$: In this case, $e$ can not get the new color $c_e$. For the contradiction, suppose that $e$ has the new color $c_e$. Therefore, $L_1^u$ will extend to a longer rainbow path with length $x + 1$ which starts from $v$. Moreover, we are looking for color connected coloring, thus the extended path has not any common color with $L_1^v$. Since $x + 1 > j$, it leads to a contradiction to the assumption that $L_2^v$ is the longest path such that $L_1^v \cup L_2^v$ is rainbow. Thus, $e$ cannot have a new color $c_e$. Hence, the optimal solution for this case of $T_e$ is at most $D(u, x, y)$. Let $c(e)$ = $c(e_1)$, then any rainbow path starting from $e$ in $T_e$ has length less than or equal to $L_1^v$ and has a common color with $L_1^v$. Therefore, the optimal solution for this case of $T_e$ is exactly $D(u, x, y)$.
    
    \item $i \leq x$ and $y < j$: In this case, $c(e) \in c(E(L_1^u))$, otherwise the concatenation of $e$ and $L_1^u$ creates a rainbow path of length $x + 1$ which is larger than length of $L_1^v$. Hence, $e$ must have the same color as the first edge of the path $L_1^u$ starting from $u$, since the coloring is color connected. Therefore, the optimal solution for this case of $T_e$ is $D(u, x, y)$.
    
    \item $i \leq x$ and $j \leq y < i$: In this case, $c(e) \in c(E(L_1^u))$, otherwise the concatenation of $e$ and $L_1^u$ creates a rainbow path longer than $L_1^v$, a contradiction. Let suppose $e_3$ be the first edge of the path $L_1^u$ which is incident to $u$.
    Hence, $e$ must have the same color as $e_3$, since we are looking for a color connected coloring. In addition, $e$ must have the same color as $e_1$, otherwise, $L_2^u$ extends to a rainbow path of length $y + 1$ which is longer that $L_2^v$ and it does not have any common color with $L_1^v$, a contradiction. Hence, $e$, $e_1$, and $e_3$ must have the same color. We have counted the color of $e_1$ as a distinct color before. On the other hand, we count the color of $e_3$ in the calculation of $D(u, x, y)$. Therefore, we have to subtract it by one to avoid duplication. Hence, the optimal solution for this case is at most $D(u, x, y) - 1$. Consider the coloring of $T_u$ that results $D(u, x, y)$ distinct colors. Let us recolor all edges in $T_u[c(e_3)]$ by $c(e_1)$. Also, let $c(e) = c(e_1)$. Length of the longest rainbow path starting from $v$ in $T_e$ in the proposed coloring is $y + 1$ which is not more than $i$. Furthermore, all rainbow paths starting from $v$ in $T_e$ have a common color with $L_1^v$, hence they do not violate the definition of $L_2^v$. Therefore, the optimal solution for this case of $T_e$ is exactly $D(u, x, y) - 1$.
    
    \item $i \leq x$ and $i \leq y$: In this case, as $i<y+1$, at least one of the $L_1^u\cup\{e\}$ or $L_2^u\cup\{e\}$ is a longer rainbow path than $L_1^v$, a contradiction to the choice of $L_1^v$. Therefore, this case is not possible and does not take part in the calculation of the value of the $D(v,i,j)$.
\end{enumerate}

Notice that we only defined $D(v, i , j)$ for $i + j < k$. Hence, by  \Cref{lem:longestpaths}, our coloring for every $D(v, i, j)$ is $P_k$-free color connected coloring.

\begin{claim}\label{clm:listselection}
	Let $A, B, C$ be three arrays of length $n$. There is an $O(n)$ algorithm for finding $\max \{A_s + B_t + \sum_{r \in [n] \setminus \{s, t\}} C_r| s \neq t , \{s,t\} \subseteq [n]\}$.
\end{claim}

\begin{ClaimProof}
First we define two arrays $A', B'$ of length $n$ as follow:
\begin{equation}
A'_i = A_i - C_i    \quad   \forall_{1 \leq i \leq n} \quad\quad
B'_i = B_i - C_i    \quad   \forall_{1 \leq i \leq n}
\end{equation}
Now the problem is reduced to finding $m = \max \{A'_i + B'_j \big\vert i \neq j\}$ in $O(n)$, since $m + \sum_{i=1}^{n}C_i$ is equal to $\max \{A_s + B_t + \sum_{r \in [n] \setminus \{s, t\}} C_r| s \neq t , \{s,t\} \subseteq [n]\}$. Let $L, R$ be two arrays of length $n$ such that $L_i = \max\{ B'_j \big\vert 1 \leq j \leq i\}$, $R_i = \max \{B'_j \big\vert i \leq j \leq n \}$. For $i = 1$, we have $L_i = B'_1$. For $i > 1$, we can obtain $L_i$ by iterating from $2$ to $n$ and calculating $L_i = \max \{L_{i-1}, B'_i\}$. Similarly, we can obtain $R_i$ by iterating from $n$ to $1$. This can be done in $O(n)$. Now we should find $\max \{ A'_i + \max\{L_{i-1}, R_{i+1}\} \big\vert 1 \leq i \leq n \}$ which can be done in $O(n)$ by checking all possible values of $i$. 
\end{ClaimProof}

According to the previous cases, we can compute $A_{u, z}$ for all $z \in [3]$ and $u \in C(v)$ in $O(k^2)$. Moreover, by \Cref{equ:jgre0}, \Cref{equ:jeq0}, and the above claim we can compute $D(v, i, j)$ in $O(deg(v))$, if we use dynamic programming approach. Therefore, the total time complexity of our algorithm is $O(|V(T)| k^4)$, since there are $O(|V(T)| k^2)$ values of $D(\cdot)$ that we need to compute.

\end{proof}

\section{Conclusions and Open Problems}
\label{sec:conclusions}

We studied the complexity of computing the anti-Ramsey
number for simple paths. We proved that computing the
$\textrm{ar}(G,P_k)$ is hard for every constant integer
$k>2$, and for $k=3$, the problem is hard to
approximate to a factor of $n^{- 1/2 - \epsilon}$.  To analyze the
exact complexity of the problem, we provided a
fine grain reduction, for a slight variation of it.  
 It remains unanswered whether the inapproximability result extends to all paths of length at least $3$.

On the positive side, we provided a linear time algorithm for
trees. Color connected coloring
does not extend to bounded treewidth graphs. However, we believe our
techniques can be employed
to provide an approximation algorithm for these
graphs. We covered paths in depth, another natural class of graphs to
be considered might
be complete graphs or cycles.

\bibliographystyle{abbrv}
 \bibliography{literature_survey}
\clearpage
\appendix

\end{document}